\let\oldvec\vec
\documentclass{svproc}
\let\vec\oldvec
\usepackage[T1]{fontenc}
\usepackage[latin9]{inputenc}
\usepackage{units}
\usepackage{amsmath}

\usepackage{amsthm}
\usepackage{amssymb}
\usepackage{graphicx}

\usepackage{url}

\usepackage[
    unicode=true,
    bookmarks=true,
    bookmarksnumbered=true,
    bookmarksopen=true,
    bookmarksopenlevel=1,
    breaklinks=false,
    pdfborder={0 0 0},
    pdfborderstyle={},
    backref=false,
    colorlinks=false]
{hyperref}

\hypersetup{
    pdftitle={Entropy on Spin Factors},
    pdfauthor={Peter Harremo{\"e}s},
    pdfpagelayout=OneColumn, 
    pdfnewwindow=true, 
    pdfstartview=XYZ, 
    plainpages=false
}

\makeatletter

\newtheorem{thm}{Theorem}
\newtheorem{defn}[thm]{Definition}
\newtheorem{lem}[thm]{Lemma}
\newtheorem{prop}[thm]{Proposition}
\newtheorem{cor}[thm]{Corollary}

\usepackage{epsfig}

\usepackage{pgf}
\usepackage{tikz}
\usepackage{mathrsfs}
\usetikzlibrary{arrows}

\makeatother

\begin{document}
\mainmatter

\title{Entropy of Spin Factors}
\author{Peter Harremo{\"e}s} 

\institute{
    Niels Brock, Copenhagen Business College, Copenhagen,\\
   WWW home page: \texttt{http:peter.harremoes.dk},\\
     \email{harremoes@ieee.org},\\
     ORCID: 0000-0002-0441-6690
}

\maketitle

\begin{abstract}
Recently it has been demonstrated that the Shannon entropy or the
von Neuman entropy are the only entropy functions that generate a
local Bregman divergences as long as the state space has rank 3 or
higher. In this paper we will study the properties of Bregman divergences
for convex bodies of rank 2. The two most important convex bodies of
rank 2 can be identified with the bit and the qubit. We demonstrate
that if a convex body of rank 2 has a Bregman divergence that satisfies
sufficiency then the convex body is spectral and if the Bregman divergence
is monotone then the convex body has the shape of a ball. A ball can
be represented as the state space of a spin factor, which is the most
simple type of Jordan algebra. We also study the existence of recovery
maps for Bregman divergences on spin factors. In general the convex
bodies of rank 2 appear as faces of state spaces of higher rank. Therefore
our results give strong restrictions on which convex bodies could
be the state space of a physical system with a well-behaved entropy
function.
\end{abstract}

\section{Introduction}

Although quantum physics has been around for more than a century the
foundation of the theory is still somewhat obscure. Quantum theory
operates at distances and energy levels that are very far from everyday
experience and much of our intuition does not carry over to the quantum
world. Nevertheless, the mathematical models of quantum physics
have an impressive predictive power. These years many scientists try
to contribute to the development of quantum computers and it becomes
more important to pinpoint the nature of the quantum resources that
may speed up the processing of a quantum computer compared with a
classic computer. There is also an interest in extending quantum physics
to be able to describe gravity on the quantum level and maybe the
foundation of quantum theory has to be modified in order to be able
to describe gravity. Therefore the foundation of quantum theory is
not only of philosophical interest but it is also important for application
of the existing theory and for extending the theory.

A computer has to consist of some components and the smallest component
must be a memory cell. In a classical computer each memory cell can
store one bit. In a quantum computer the memory cells can store one
qubit. In this paper we will focus on such minimal memory cells and
demonstrate that under certain assumptions any such memory cell can
be represented as a so-called spin factor. We formalize the memory
cell by requiring that the state space has rank 2. In some recent
papers it was proved that a local Bregman divergence on a state
space of rank at least 3 is proportional to information divergence
and the state space must be spectral \cite{Harremoes2017,Harremoes2017a}.
Further, on a state space of rank at least 3 locality of a Bregman divergence
is equivalent to the conditions called sufficiency and monotonicity.
If the rank of the state space is 2 the situation is quite different.
First of all the condition called locality reduce almost to a triviality.
Therefore it is of interest to study sufficiency and monotonicity
on state spaces of rank 2. 

The paper is organized as follows. In the first part we study convex
bodies and use mathematical terminology without reference to physics.
The convex bodies may or may not correspond to state spaces of physical
systems. I Section 2 some basic terminology regarding convex sets
is established and the rank of a set is defined. In Section 3 regret
and Bregman divergences are defined, but for a detailed motivation
we refer to \cite{Harremoes2017}. In Section 4 spectral sets are
defined and it is proved that a spectral set of rank 2 has central symmetry. In Section 5 sufficiency of a regret function is defined
and it is proved that a convex body of rank 2 with a regret function
that satisfies sufficiency is spectral. 

Spin factors are introduced in Section 6. Spin factors appear as sections
of state spaces of physical systems described by density matrices
on complex Hilbert spaces. Therefore we will borrow some terminology
from physics. In Section 7 monotonicity of a Bregman divergence is
introduced. It is proved that a convex body with a sufficient Bregman
divergence that is monotone under dilations can be represented as a spin factor. For general spin factors we have not obtained a simple characterization of the monotone Bregman, but some partial results are presented in Section 8. In Section 9 it is
proved that equality in the inequality for a monotone Bregman divergence
implies the existence of a recovery map.

In this paper we focus on finite dimensional convex bodies.
Many of the results can easily be generalized to bounded convex set
in separable Hilbert spaces, but that woulds require that topological considerations
are taken into account.

\section{Convex Bodies of Rank 2}
In this paper we will work within a category where the objects are
\emph{convex bodies}, i.e. finite dimensional convex compact sets.
The morphisms will be \emph{affinities}, i.e. affine maps between
convex bodies. The convex bodies are candidates for state spaces of
physical systems, so a point in a convex bodies might be interpreted
as a state that may represent our knowledge of the physical system.
A convex combination $\sum p_{i}\cdot\sigma_{i}$ is interpreted as
a state where the system is prepared in state $\sigma_{i}$ with probability
$p_{i}.$ In classical physics the state space is a simplex and in
the standard formalism of quantum physics the state space is isomorphic
to the density matrices on a complex Hilbert space.

A bijective affinity will be called an \emph{isomorphism}. Let $\mathcal{K}$
and $\mathcal{L}$ denote convex bodies. An affinity $S:\mathcal{K}\to\mathcal{L}$
is called a \emph{section} if there exists an affinity $R:\mathcal{L}\to\mathcal{K}$
such that $R\circ S=id_{K},$ and such an affinity $R$ is called
a \emph{retraction}. Often we will identify a section $S:\mathcal{K}\to\mathcal{L}$
with the set $S\left(\mathcal{K}\right)$ as a subset of $\mathcal{L}.$
Note that the affinity $S\circ R:\mathcal{L}\to\mathcal{L}$ is \emph{idempotent}
and that any idempotent affinity determines a section/retraction pair.
We say that $\sigma_{0}$ and $\sigma_{1}$ are \emph{mutually singular}
if there exists a section $S:\left[0,1\right]\to K$ such that $S\left(0\right)=\sigma_{0}$
and $S\left(1\right)=\sigma_{1}.$ Such a section is illustrated on
Figure \ref{fig:section}. A retraction $R:K\to\left[0,1\right]$
is a special case of a \emph{test} \cite[p. 15]{Holevo1982} (or an
\emph{effect} as it is often called in generalized probabilistic
theories \cite{Alfsen2003}). We say that $\sigma_{0},\sigma_{1}\in\mathcal{K}$
are \emph{orthogonal} if $\sigma_{0}$ and $\sigma_{1}$ belong to
a face $\mathcal{F}$ of $\mathcal{K}$ such that $\sigma_{0}$ and
$\sigma_{1}$ are mutually singular in $\mathcal{F}.$ 

\begin{figure}[tbh]
\begin{centering}
\begin{tikzpicture}[line cap=round,scale=3, line join=round,>=triangle 45,x=1.0cm,y=1.0cm] 
\clip(-1.21875,-1.23) rectangle (1.24375,0.65125); 
 
\draw [rotate around={21.136844503046866:(0.,0.)},line width=1.2pt,color=black,fill=gray!30!white,fill opacity=0.5] (0.,0.) ellipse (1.0564005217811827cm and 0.4733053322966527cm); 
\draw [line width=1.2pt,color=blue] (-0.9982529445650812,-0.3290184740952256)-- (0.9991413137792894,0.32045854181739175); 
\draw (-1.20,-0.2175) node[anchor=north west] {$\sigma_0$}; 
\draw (1.075,0.42) node[anchor=north west] {$\sigma_1$}; 
\draw [fill=black] (-0.9982529445650812,-0.3290184740952256) circle (1pt); 
\draw [fill=black] (0.9991413137792894,0.32045854181739175) circle (1pt); 

\draw [line width=1.2pt,color=blue] (-1.,-1.)-- (1.,-1.);
\draw[color=blue] (-1.0125,-1.1268749999999976) node {0}; 
\draw[color=blue] (0.98125,-1.1081249999999976) node {1}; 

\draw [->,line width=1.2pt,dash pattern=on 6pt off 6pt] (-0.9982529445650812,-0.32901847409522555) --  (-1.,-1.) ; 
\draw [->,line width=1.2pt,dash pattern=on 6pt off 6pt] (0.9991413137792894,0.32045854181739175) -- (1.,-1.); 

\end{tikzpicture}
\par\end{centering}
\caption{\label{fig:section}A retraction with orthogonal points $\sigma_{0}$
and $\sigma_{1}$. The corresponding section is obtained by reversing
the arrows. }
\end{figure}
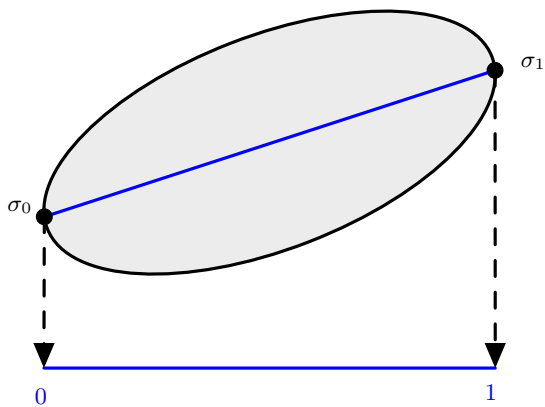

The following result was stated in \cite{Harremoes2017b} without
a detailed proof.
\begin{thm}
\label{thm:SectionGennemPunkt}If $\sigma$ is a point in a convex
body $\mathcal{K}$ then $\sigma$ can be written as a convex combination
$\sigma=\left(1-t\right)\cdot\sigma_{0}+t\cdot\sigma_{1}$ where $\sigma_{0}$
and $\sigma_{1}$ orthogonal.
\end{thm}

\begin{proof}
Without loss of generality we may assume that $\sigma$ is an algebraically
interior point of $\mathcal{K}.$ For any $\sigma_0$ on the boundary of
$\mathcal{K}$ there exists a $\sigma_1$ on the boundary of $\mathcal{K}$
and $t_{\sigma_0}\in]0,1[$ such that $\left(1-t_{\sigma_0}\right)\cdot\sigma_0 +t_{\sigma_0}\cdot\sigma_1 =\sigma.$
Let $R$ denote a retraction $R:\mathcal{K}\to\left[0,1\right]$ such that $R\left(\sigma_{0}\right)=0.$ 
Let $S$
denote a section corresponding to $R$ such that $S\left(0\right)=\sigma_{0}.$ Let $\pi_1$ denote the point $S(1)$.
\begin{figure}[tbh]
\begin{centering}
\begin{tikzpicture}[line cap=round,line join=round,>=triangle 45,x=5cm,y=5cm]
\clip(-0.1,-0.13522696734534642) rectangle (1.2247072731364657,1.4027796191742954);
\draw [rotate around={-52.01812173396323:(0.5,0.75)},line width=1.2pt,color=black,fill=gray!30!white,fill opacity=0.5] (0.5,0.75) ellipse (0.6403882032022076 and 0.3903882032022076);
\draw [->,line width=1.2pt,dash pattern=on 6pt off 6pt,color=black] (0.8392847490027451,0.9476257279716631)-- (0.8392847490027451,0);
\draw [line width=1.2pt,color=blue] (1,0)-- (0,0);
\draw [line width=1.2pt,color=blue] (0,1)-- (1,0.5);
\draw [line width=1.2pt,color=black] (0,1)-- (0.8392847490027451,0.9476257279716631);
\draw [line width=1.2pt,color=red] (1,0.5)-- (0.1278654304850012,1.2700071927827106);
\draw [->,line width=1.2pt,dash pattern=on 6pt off 6pt,color=black] (0,1)-- (0,0);
\draw [->,line width=1.2pt,dash pattern=on 6pt off 6pt,color=red] (1,0.5)-- (1,0);
\draw (0.88,1) node[anchor=north west] {$\sigma_1$};
\draw (0.4406254839303717,1.1) node[anchor=north west] {$\sigma$};
\draw (-0.12,1.05) node[anchor=north west] {$\sigma_0$};
\draw [color=black](0,-0.04) node[anchor=north] {0};
\draw [color=red](1,-0.04) node[anchor=north] {1};
\draw [color=red](1.039,0.52) node[anchor=north west] {$\pi_1$};
\draw [color=red](0.01,1.34) node[anchor=north west] {$\pi_0$};
\draw [->,line width=1.2pt,dash pattern=on 6pt off 6pt,color=red] (0.1278654304850012,1.2700071927827106)-- (0.1278654304850012,0);

\draw [color=red](0.008375266803935352,-0.04) node[anchor=north west] {$R\left(\pi_0\right)$};
\draw [color=black](0.7019860803324031,-0.04) node[anchor=north west] {$R\left(\sigma_1\right)$}; 
\draw [fill=red] (1,0.5) circle (2pt);
\draw [fill=black] (0,1) circle (2pt);
\draw [fill=red] (0.1278654304850012,1.2700071927827106) circle (2pt);
\draw [fill=black] (0.8392847490027451,0.9476257279716631) circle (2pt);
\draw  [fill=black] (0.46666832240342976,0.970878282140112) circle (2pt);
\end{tikzpicture}
\par\end{centering}
\caption{\label{fig:ulighed}Illustration to the proof of Theorem \ref{thm:SectionGennemPunkt}}
\end{figure}
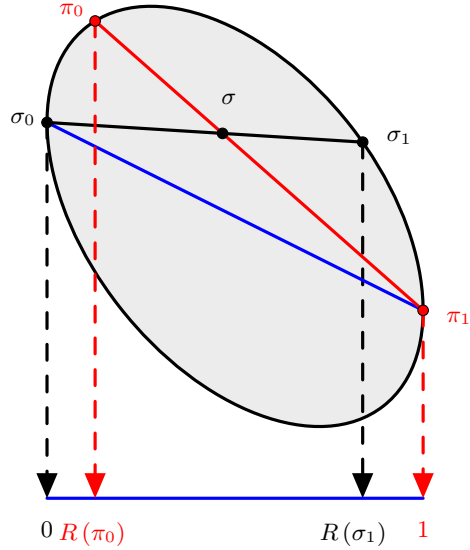
There exists a point $\pi_0$ on the boundary such that $\sigma=\left(1-t_{\pi_0}\right)\cdot\pi_0 + t_{\pi_0}\cdot \pi_1$. Then
\begin{align}
R\left(\sigma\right)&=R\left(\left(1-t_{\pi_0}\right)\cdot\pi_0 + t_{\pi_0}\cdot \pi_1\right)\\
&=\left(1-t_{\pi}\right)\cdot R\left(\pi_0 \right) + t_{\pi}\cdot R\left(\pi_1\right)\\
&\geq t_{\pi}
\end{align}
and 
\begin{align}
R\left(\sigma\right)&=R\left(\left(1-t_{\sigma_0}\right)\cdot\sigma_{0} + t_{\sigma_0}\cdot \sigma_{1}\right)\\
&=\left(1-t_{\sigma_0}\right)\cdot 0 + t_{\sigma_0}\cdot R\left(\sigma_{1}\right)\\
&=t_{\sigma_0}\cdot R\left(\sigma_{1}\right).\\
\end{align}
Therefore 
\begin{equation}\label{minimal}
t_{\sigma_0}\cdot R\left(\sigma_{1}\right)\geq t_{\pi_0}
\end{equation}
Since $t_{\sigma_0}$ is a continuous function of $\sigma_0$ the function we may choose $\sigma_0$ such that $t_{\sigma_0}$
 is minimal, but if $t_{\sigma_0}$ is minimal Inequality (\ref{minimal}) implies that $R\left(\sigma_1\right)=1$ so that $\sigma_0$
 and $\sigma_1$ are orthogonal.
\end{proof}

Iterated use of Theorem \ref{thm:SectionGennemPunkt} leads to an extended
version of Caratheodory's theorem \cite[Thm. 2]{Harremoes2017b}.
\begin{thm}[Orthogonal Caratheodory Theorem]
Let $\mathcal{K}$ denote a convex body of dimension $d$. Then any
point $\sigma\in\mathcal{K}$ has a decomposition $\sigma=\sum_{i=1}^{n}t_{i}\cdot\sigma_{i}$
where $t_{1}^{n}$ is a probability vector and $\sigma_{i}$ are orthogonal
extreme points in $\mathcal{K}$ and $n\leq d+1\, .$ 
\end{thm}
The Caratheodory number of a convex body is the maximal number of extreme points needed to decompose a point into extreme points. We need a similar definition related to orthogonal decompositions.
\begin{defn}
The\emph{ rank} of a convex body $\mathcal{K}$ is the maximal number
of orthogonal extreme points needed in an orthogonal decomposition
of a point in $\mathcal{K}.$
\end{defn}

If $\mathcal{K}$ has rank 1 then it is a singleton. Some examples of convex
bodies of rank 2 are illustrated in Figure \ref{fig:ConvexBodiesRank2}. Clearly the Caratheodory number lower bounds the rank of a convex body. 
Figure \ref{fig:CentrallySymmetricWithFace} provides an example where the Carathodory number is different from the rank.  The rest of this
paper will focus on convex bodies of rank 2.
Convex bodies of rank 2 satisfy \emph{weak spectrality} as defined
in \cite{Barnum2015}. 

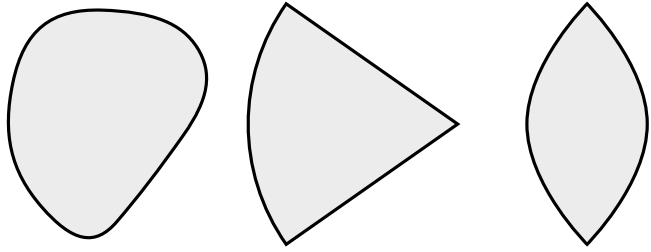
\begin{figure}[tbh]
\begin{centering}

\ifx\JPicScale\undefined\def\JPicScale{0.8}\fi
\unitlength \JPicScale mm
\begin{tikzpicture}[x=\unitlength,y=\unitlength,inner sep=0pt]

\draw [line width=1.2pt,fill=gray!30!white,fill opacity=0.5]
(24.61,57.04) .. controls (22.76,48.43) and (24.18,42.19) .. 
(29.34,36.25) .. controls (34.5,30.31) and (38.25,29.56) .. 
(41.84,33.75) .. controls (45.43,37.94) and (48.77,42.28) .. 
(52.96,48.22) .. controls (57.15,54.16) and (57.84,58.49) .. 
(55.26,62.63) .. controls (52.68,66.78) and (47.93,68.67) .. 
(39.41,68.95) .. controls (30.89,69.22) and (26.45,65.65) .. 
cycle;

\draw [line width=1.2pt,fill=gray!30!white,fill opacity=0.5](70,70) arc (144.94:215.06:34.82) -- (98.5,50) -- cycle;

\draw [line width=1.2pt,fill=gray!30!white,fill opacity=0.5]
(110,50) .. controls (110,40) and (120,30) .. 
(120,30) .. controls (120,30) and (130,40) .. 
(130,50) .. controls (130,60) and (120,70) .. 
(120,70) .. controls (120,70) and (110,60) .. 
cycle;
\end{tikzpicture}

\par\end{centering}
\caption{\label{fig:ConvexBodiesRank2}Convex bodies of rank 2. The convex
body to the left has a smooth strictly convex boundary so that any
point on the boundary has exactly one orthogonal point. The body in
the middle has a set of three extreme points that are orthogonal, but any
point can be written as a convex combination of just two points. The
convex body to the right is centrally symmetric without 1-dimensional
proper faces, i.e. it is a spectral set.}
\end{figure}

If $\mathcal{K}$ is a convex body it is sometimes convenient to consider
the cone $\mathcal{K}_{+}$ generated by $\mathcal{K}$ . The cone
$\mathcal{K}_{+}$ consist of elements of the form $x\cdot\sigma$
where $x\geq0$ and $\sigma\in\mathcal{\mathcal{K}}$. Elements of
the cone are called \emph{positive elements} and such elements can
be multiplied by positive constants via $x\cdot\left(y\cdot\sigma\right)=\left(x\cdot y\right)\cdot\sigma$
and can be added as follows. 
\begin{equation}
x\cdot\rho+y\cdot\sigma=\left(x+y\right)\cdot\left(\frac{x}{x+y}\cdot\rho+\frac{y}{x+y}\cdot\sigma\right).
\end{equation}
For a point $\sigma\in\mathcal{K}$ the \emph{trace} of $x\cdot\sigma\in\mathcal{K}_{+}$
is defined by $\mathrm{tr}\left[x\cdot\sigma\right]=x.$ The cone
$\mathcal{K}_{+}$ can be embedded in a real vector space by taking
the affine hull of the cone and use the apex of the cone as origin
of the vector space and the trace extends to a linear function on
this vector space. In this way a convex body $\mathcal{K}$ can
be identified with the set of positive elements in a vector
space with trace 1. 
\begin{lem}
Let $\mathcal{K}$ be a convex body and let $\Phi:\mathcal{K}\to\mathcal{K}$
be an affinity. Let $\Phi_{\mu}=\sum_{n=0}^{\infty}\frac{\mu^{n}}{n!}\mathrm{e}^{-\mu}\Phi^{\circ n}$.
Then $\Phi_{\infty}=\lim_{\mu\to\infty}\Phi_{\mu}$ is a retraction
of $\mathcal{K}$ onto the set of fix-points of $\Phi.$ 
\end{lem}

\begin{proof}
Since $\mathcal{K}$ is compact the affinity $\Phi$ has a fix-point that
we will call $s_{0}$. The affinity can be extended to a positive trace
preserving affinity of the real vector space generated by $\mathcal{K}$
into itself. Since $\Phi$ maps a convex body into itself all the
eigenvalues of $\Phi$ are numerically upper bounded by 1. The affinity
can be extended to a complexification of the vector space. On this
complexification of the vector space there exist a basis in which
the affinity $\Phi$ has the Jordan normal form with blocks
of the form
\begin{equation}
\left(\begin{array}{cccccc}
\lambda & 1 & 0 & \cdots & 0 & 0\\
0 & \lambda & 1 & \ddots & \ddots & \ddots\\
0 & 0 & \lambda & \ddots & \ddots & \ddots\\
\vdots & \vdots & \vdots & \ddots & \ddots & 0\\
0 & 0 & 0 & \cdots & \lambda & 1\\
0 & 0 & 0 & \cdots & 0 & \lambda
\end{array}\right)
\end{equation}
and $\Phi^{n}$ has blocks of the form
\begin{equation}
\left(\begin{array}{cccccc}
\lambda^{n} & \binom{n}{n-1}\lambda^{n-1} & \binom{n}{n-2}\lambda^{n-2} & \ddots & \binom{n}{n-\ell+2}\lambda^{n-\ell+2} & \binom{n}{n-\ell+1}\lambda^{n-\ell+1}\\
0 & \lambda^{n} & \binom{n}{1}\lambda^{n-1} & \ddots & \ddots & \ddots\\
0 & 0 & \lambda^{n} & \ddots & \ddots & \ddots\\
\vdots & \vdots & \vdots & \ddots & \ddots & \ddots\\
0 & 0 & 0 & \cdots & \lambda^{n} & \binom{n}{n-1}\lambda^{n-1}\\
0 & 0 & 0 & \cdots & 0 & \lambda^{n}
\end{array}\right).
\end{equation}
Now 
\begin{align}
\sum_{n=0}^{\infty}\frac{\mu^{n}}{n!}\mathrm{e}^{-\mu}\binom{n}{n-j}\lambda^{n-j} & =\sum_{n=j}^{\infty}\frac{\mu^{n}}{\left(n-j\right)!j!}\mathrm{e}^{-\mu}\lambda^{n-j}\\
 & =\frac{\mu^{j}}{j!}\sum_{n=j}^{\infty}\frac{\mu^{n-j}}{\left(n-j\right)!}\mathrm{e}^{-\mu}\lambda^{n-j}
\end{align}
tends to zero for $\mu$ tending to infinity except if $\lambda=1$.
If $\lambda=1$ then there is no uniform upper bound on $\Phi^{n}$
except if the Jordan block is diagonal. Therefore $\Phi_{\mu}=\sum_{n=0}^{\infty}\frac{\mu^{n}}{n!}\mathrm{e}^{-\mu}\Phi^{\circ n}$
convergences to a map $\Phi_{\infty}$ that is diagonal with eigenvalues
0 and 1, i.e. a idempotent. Since $\Phi$ and $\Phi_{\infty}$ commute
they have the same fix-points.
\end{proof}
\begin{prop}
Let $\Phi$ denote an affinity $\mathcal{K}\to\mathcal{L}$ and let
$\Psi$ denote an affinity $\mathcal{L}\to\mathcal{K}.$ Then the
set of fix-points of $\Psi\circ\Phi$ is a section of $\mathcal{K}$
and the set of fix-points of $\Phi\circ\Psi$ is a section of $\mathcal{L}$.
The affinities $\Phi$ and $\Psi$ restricted to the fix-point sets
are isomorphisms between these sets.
\end{prop}

\section{Regret and Bregman Divergences}

Consider a payoff function where the payoff may represent extracted
energy or how much data can be compressed or something else. Our payoff
depends both of the state of the system and of some choice that we
can make. Let $F\left(\sigma\right)$ denote the maximal mean payoff
when our knowledge is represented by $\sigma.$ Then $F$ is a convex
function on the convex body. 

The two most important examples are squared the Euclidean norm squared $F\left(\vec{v}\right)=\left\Vert \vec{v}\right\Vert _{2}^{2}$ 
defined on a vector space and minus the von Neuman entropy $F\left(\sigma\right)=\mathrm{tr}\left[\sigma\ln\left(\sigma\right)\right].$  
Note that Shannon entropy may be considered as a special case of von Neuman entropy when all operators commute. 
We may also consider $F\left(\sigma\right)=-S_{\alpha}\left(\sigma\right)$ where the Tsallis entropy of order $\alpha>0$ is defined by

\begin{equation}
S_{\alpha}\left(\sigma\right)=-\mathrm{tr}\left[\sigma\log_{\alpha}\left(\sigma\right)\right]
\end{equation}
and where the logarithm of order $\alpha \neq 1$ is given by
\begin{equation}
\log_{\alpha}\left(x\right)=\frac{x^{\alpha -1}-1}{\alpha -1} \, 
\end{equation}
and $\log_{1}\left(x\right)=\ln\left(x\right)$ . 
We will study such entropy functions via the corresponding regret functions that are defined by:

\begin{defn}
Let $F$ denote a convex function defined on a convex body $\mathcal{K}$.
For $\rho,\sigma\in\mathcal{K}$ we define the \emph{regret function}
$D_{F}$ by
\begin{equation}
D_{F}\left(\rho,\sigma\right)=F\left(\rho\right)-\left(F\left(\sigma\right)+\lim_{t\to0_{+}}\frac{F\left(\left(1-t\right)\cdot\sigma+t\cdot\rho\right)-F\left(\sigma\right)}{t}\right).\label{eq:Regret}
\end{equation}
The regret function $D_F$ is \emph{strict} if $D_{F}\left(\rho,\sigma\right)=0$ implies that $\rho =\sigma$ . If $F$ is differentiable
the regret function is called a \emph{Bregman divergence}.
\end{defn}

The interpretation of the regret function is that $D_{F}\left(\rho,\sigma\right)$
tells how much more payoff one could have obtained if the state is $\rho$ but one act as if the
state was $\sigma$. This is illustrated
in Figure \ref{fig:regret}.

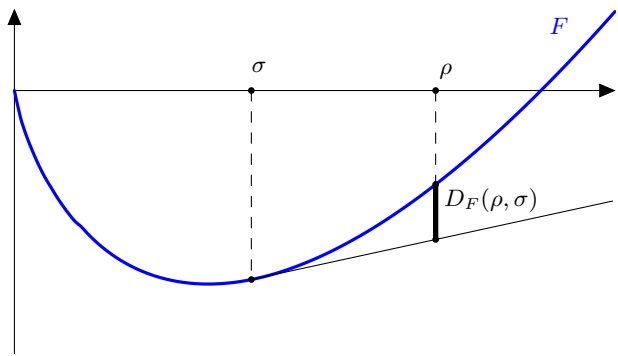
\begin{figure}[tbh]
\begin{centering}
\begin{tikzpicture}[scale=0.70,line cap=round,line join=round,>=triangle 45,x=10.0cm,y=10.0cm] 
\draw[->,color=black] (0.,0.) -- (1.1419478737997253,0.); 
\draw[->,color=black] (0.,-0.5) -- (0.,0.15580246913580265); 
\clip(-0.06431,-0.5008779149519884) rectangle (1.1419478737997253,0.15580246913580265); 
\draw[line width=1.2pt,color=blue,smooth,samples=100,domain=0:1.1419478737997253] plot(\x,{(\x)*ln((\x)+1.0E-4)}); 
 
\draw (0.45,-0.35922847440746253) -- (1.1362, -0.21); 
\draw [dash pattern=on 4pt off 4pt] (0.8,0.)-- (0.8,-0.178414847300847); 
 
\draw [line width=2.pt] (0.8,-0.178414847300847)-- (0.8,-0.283); 
\draw [dash pattern=on 4pt off 4pt] (0.45,0.)-- (0.45,-0.35922847440746253); 
\draw (0.4343468211923174,0.07152263374485619) node[anchor=north west] {$\sigma$}; 
\draw (0.790781098312178,0.0680109739369001) node[anchor=north west] {$\rho$}; 
\draw (0.8,-0.17) node[anchor=north west] {$D_{F}(\rho ,\sigma)$}; 
\draw (1,0.1540466392318246) node[anchor=north west, color=blue] {$F$}; 

\draw [fill=black] (0.8,-0.178414847300847) circle (1.5pt); 
\draw [fill=black] (0.45,-0.35922847440746253) circle (1.5pt); 
\draw [fill=black] (0.8,-0.283) circle (1.5pt); 
\draw [fill=black] (0.45,0.) circle (1.5pt); 
\draw [fill=black] (0.8,0.) circle (1.5pt); 

\end{tikzpicture}
\par\end{centering}
\caption{The regret equals the vertical distance between the curve and the
tangent.\label{fig:regret}}
\end{figure}

The two most important examples of Bregman divergences are squared Euclidean distance $\left\Vert \vec{v}-\vec{w}\right\Vert _{2}^{2}$
that is generated by the squared Euclidean norm and information divergence 
\begin{equation}
D\left(\rho\Vert\sigma\right)=\mathrm{tr}\left[\rho\left(\ln\left(\rho\right)-\ln\left(\sigma\right)\right)-\rho+\sigma\right]
\end{equation}
that is generated by minus the von Neuman entropy. The Bregman divergence generated by $-S_{\alpha}$ is called the Bregman divergence of order $\alpha$ and is denoted $D_{\alpha}\left(\rho ,\sigma\right)$ . Various examples of payoff functions and corresponding regret functions
are discussed in \cite{Harremoes2017} where some basic properties
of regret functions are also discussed. If $F$ is differentiable
the regret function is a Bregman divergence and the
formula (\ref{eq:Regret}) reduces to
\begin{equation}
D_{F}\left(\rho,\sigma\right)=F\left(\rho\right)-\left(F\left(\sigma\right)+\left\langle \nabla F\left(\sigma\right)\mid \rho-\sigma \right\rangle \right).
\end{equation}
Bregman divergences were introduced in \cite{Bregman1967}, but they
only gained popularity after their properties were investigated i
great detail in \cite{Banerjee2005}. A Bregman divergence satisfies the \emph{Bregman
equation}
\begin{equation}
\sum t_{i}\cdot D_{F}\left(\rho_{i}, \sigma\right)=\sum t_{i}\cdot D_{F}\left(\rho_{i}, \bar{\rho}\right)+D_{F}\left(\bar{\rho}, \sigma\right)
\end{equation}
 where $\left(t_{1},t_{2},\dots\right)$ is a probability vector and
$\bar{\rho}=\sum t_{i}\cdot\rho_{i}$ .

Assume that $D_{F}$ is a Bregman divergence on the convex body $\mathcal{K}.$ 
If the state is not know exactly but we know that $s$ is one of the
states $s_{1},s_{2},\dots,s_{n}$ then the \emph{minimax regret} is
defined as 
\begin{equation}
C_{F}=\inf_{\sigma\in\mathcal{K}}\sup_{\rho\in\mathcal{K}}D_{F}\left(\rho,\sigma\right).
\end{equation}
The point $\sigma$ that achieves the minimax regret will be denoted
by $\sigma_{opt}.$ 

\begin{thm}
If $\mathcal{K}$ is a convex body with a Bregman divergence $D_{F}$ and with a probability vector $\left(t_{1},t_{2},\dots,t_{n}\right)$ 
on the points $\rho_{1},\rho_{2},\dots,\rho_{n}$ with $\bar{\rho}=\sum t_{i}\cdot\rho_{i}$
and $\sigma_{opt}$ achieves the minimax regret then 
\begin{equation}\label{eq:redundansulighed}
C_{F}\geq\sum t_{i}\cdot D_{F}\left(\rho_{i},\bar{\rho}\right)+D_{F}\left(\bar{\rho},\sigma_{opt}\right).
\end{equation}
\end{thm}

\begin{proof}

If $\sigma_{opt}$ is optimal then 
\begin{eqnarray}
C_{F} & = & \sum_{i}t_{i}\cdot C_{F}\\
 & \geq & \sum_{i}t_{i}\cdot D_{F}\left(\rho_{i},\sigma_{opt}\right)\\
 & = & \sum_{i}t_{i}\cdot D_{F}\left(\rho_{i},\bar{\rho}\right)+D_{F}\left(\bar{\rho},\sigma_{opt}\right)
\end{eqnarray}
which proves Inequality (\ref{eq:redundansulighed}).
\end{proof}

One can formulate a minimax theorem for divergence, but we will prove a result that is stronger than a minimax theorem in the sense that it gives an upper bound on how close a specific strategy is to the optimal strategy. First we need the following lemma.

\begin{lem}
Let $\mathcal{K}$ be a convex body with a Bregman divergence $D_F$ that is lower semi-continuous. Let $\mathcal{L}$ denote a closed convex subset of $\mathcal{K}$. For any $\sigma\in\mathcal{K}$ there exists a point $\sigma^{*}\in\mathcal{L}$ such that 
\begin{equation}\label{eq:pythagoras}
D_{F}\left(\rho,\sigma\right)\geq D_{F}\left(\rho,\sigma^{*}\right)+D_{F}\left(\sigma^{*},\sigma\right) 
\end{equation}
for all $\rho\in\mathcal{L}$. In particular $\sigma^*$ minimizes $D_{F}\left(\rho,\sigma\right)$ under the constraint that $\rho\in\mathcal{L}$.
\end{lem}
\begin{proof}
Using that $\mathcal{L}$ is closed and lower semicontinuity of $D_F$ we find a point $\sigma^* \in \mathcal{L}$ that minimizes $D_{F}\left(\rho,\sigma\right)$ under the constraint that $\rho\in\mathcal{L}$. Define 
\begin{equation}
\rho_t =\left(1-t\right)\cdot\sigma^* +t\cdot\rho .
\end{equation}
Then according to the Bregman equation
\begin{multline}
\left(1-t\right)\cdot D_{F}\left(\sigma^* ,\sigma\right)+t\cdot D_{F}\left(\rho,\sigma\right)
=\left(1-t\right)\cdot D_{F}\left(\sigma^* ,\rho_{t}\right)+t\cdot D_{F}\left(\rho,\rho_{t}\right)+D_{F}\left(\rho_{t},\sigma\right)\\
\geq t\cdot D_{F}\left(\rho,\rho_{t}\right)+D_{F}\left(\sigma^* ,\sigma\right)
\end{multline}
After reorganizing the terms and dividing by $t$ we get
\begin{equation}
D_{F}\left(\rho,\sigma\right) \geq D_{F}\left(\rho,\rho_{t}\right)+ D_{F}\left(\sigma^* ,\sigma\right)\, .
\end{equation}
Inequality (\ref{eq:pythagoras})  is obtained by letting $t$ tend to zero and using lower semi-continuity.
\end{proof}

\begin{thm}
If $\mathcal{K}$ is a convex body with a Bregman divergence $D_{F}$ that is lower semi-continuous in both variables and such that $F$ is continuously differentiable $C^1$. Then 
\begin{equation}\label{eq:minimax}
C_{F}=\sup_{\vec{t}}\sum_{i}t_{i}\cdot D_{F}\left(\rho_{i},\bar{\rho}\right)
\end{equation}
where the supremum is taken over all probability vectors $\vec{t}$
supported on $\mathcal{K}$. Further the following inequality holds
\begin{equation}\label{eq:maxred}
\sup_{\rho\in\mathcal{K}}D_{F}\left(\rho,\sigma\right)\geq C_{F}+D_{F}\left(\sigma_{opt},\sigma\right)
\end{equation}
for all $\sigma$.
\end{thm}

\begin{proof}
First we prove the theorem for a convex polytope $\mathcal{L}\subseteq\mathcal{K}$. 
Assume that $\rho_1 ,\rho_2 ,\dots , \rho_n$ are the extreme points of $\mathcal{L}$. 
Let $\sigma_{opt}\left(\mathcal{L}\right)$ denote a point that minimizes that $\sup_{\rho\in\mathcal{K}}D_{F}\left(\rho,\sigma\right)$. 
Let $J$ denote the set of indices $i$ for which 
\begin{equation}
D_{F}\left(\rho_i ,\sigma\right)=\sup_{\rho\in\mathcal{L}}D_{F}\left(\rho ,\sigma\right)
\end{equation}
Let $\mathcal{M}$ denote the convex hull of $\rho_i ,\, i\in J$. Let $\pi$ denote the projection of $\sigma_{opt}$ on $\mathcal{M}$. The there exists a mixture such that $\sum_{i\in J}t_{i}\cdot\rho_{i}=\pi$. 
Then for any $\sigma$
\begin{align}
\sup_{\rho\in\mathcal{L}}D_{F}\left(\rho,\sigma\right) & \geq  \sum_{i\in J}t_{i}\cdot D_{F}\left(\rho_{i},\sigma\right)\label{eq:generel}\\
 & =  \sum_{i\in J}t_{i}\cdot D_{F}\left(\rho_{i},\bar{\rho}\right)+D_{F}\left(\bar{\rho},\sigma\right).
\end{align}
Since all divergences $D_{F}\left(\rho_{i},\sigma\right)$ where $i\in J$ can be decreased by moving $\sigma$ from $\sigma_{opt}$ towards $\pi$ 
and the divergences $D_{F}\left(\rho_{i},\sigma\right)$ where $i\notin J$ are below $C\left(\mathcal{L}\right)$ as long as $\sigma$ is only moved a little towards $\pi$ we have that $\pi =\sigma_{opt}$ and that (\ref{eq:minimax}) holds. 
Inequality (\ref{eq:maxred}) follows from inequality (\ref{eq:generel}) when $\bar{\rho}=\rho_{opt}$. 

Let $\mathcal{L}_{1}\subseteq\mathcal{L}_{1}\subseteq\dots\subseteq\mathcal{K}$ denote an increasing sequence of polytopes such that the union contain the interior of $\mathcal{K}$. We have
\begin{equation}\label{eq:sekvens}
C_{F}\left(\mathcal{L}_{1}\right)   \leq       C_{F}\left(\mathcal{L}_{1}\right) \leq \dots\leq  C_{F}\left(\mathcal{K}\right)
\end{equation}
Let $\sigma_{opt,i}$ denote a point that is optimal for $\mathcal{L}_i$ . By compactness of $\mathcal{K}$ we may assume that $\sigma_{i}\to\sigma_{\infty}$ for $i\to\infty$ for some point $\sigma_{\infty}\in\mathcal{K}$. 
Otherwise we just replace the sequence by a subsequence. For any $\rho\in\mathcal{L}_i$ we have
\begin{align}
D_{F}\left(\rho , \sigma_{i}\right) &\leq \lim_{i\to\infty}\inf D_{F}\left(\rho , \sigma_{i}\right)\\
&\leq \lim_{i\to\infty} C_{F}\left(\mathcal{L}_{i}\right) \, .
\end{align}
By lower semi-continuity
\begin{equation}
D_{F}\left(\rho , \sigma_{\infty}\right)\leq \lim_{i\to\infty} C_{F}\left(\mathcal{L}_{i}\right) \, .
\end{equation}
By taking the supremum over all interior points $\rho\in\mathcal{K}$ we obtain
\begin{align}
C_{F}\left(\mathcal{K}\right)&\leq\sup_{\rho\in\mathcal{K}} D_{F}\left(\rho , \sigma_{\infty}\right)\\
&=\lim_{i\to\infty} C_{F}\left(\mathcal{L}_{i}\right)\\
&=\sup_{\vec{t}}\sum_{i}t_{i}\cdot D_{F}\left(\rho_{i},\bar{\rho}\right),
\end{align}
which in combination with (\ref{eq:sekvens}) proves (\ref{eq:minimax}) and also proves that $\sigma_{\infty}$ is optimal.
We also have 
\begin{align}
\sup_{\rho\in\mathcal{K}}D_{F}\left(\rho,\sigma\right)&=\lim_{i\to\infty}\sup_{\rho\in\mathcal{L}_{i}}D_{F}\left(\rho,\sigma\right)\\
&\geq \lim_{i\to\infty}\inf \left(C_{F}\left(\mathcal{L}_{i}\right)+D_{F}\left(\sigma_{i},\sigma\right)\right)\\
&\geq C_{F}+D_{F}\left(\sigma_{\infty},\sigma\right)
\end{align}
which proves Inequality (\ref{eq:maxred}).
\end{proof}

\section{Spectral Sets}

Let $\mathcal{K}$ denote a convex body of rank 2. Then $\sigma\in\mathcal{K}$
is said to have \emph{unique spectrality} if all orthogonal decompositions
$\sigma=\left(1-t\right)\cdot\sigma_{0}+t\cdot\sigma_{1}$ have the
same coefficients $\left\{ 1-t,t\right\} $ and the set $\left\{ 1-t,t\right\} $
is called the \emph{spectrum} of $\sigma.$ If all elements of $\mathcal{K}$
have unique spectrality we say that $\mathcal{K}$ is \emph{spectral}.
A convex body $\mathcal{K}$ is said to be \emph{centrally symmetric}
with \emph{center} $c$ if for any point $\sigma\in\mathcal{K}$ there
exists a \emph{centrally inverted} point $\tilde{\sigma}$ in $\mathcal{K}$,
i.e. a point $\tilde{\sigma}\in\mathcal{K}$ such that $\frac{1}{2}\sigma+\frac{1}{2}\tilde{\sigma}=c\, .$ 
\begin{thm}
A spectral set $\mathcal{K}$ of rank 2 is centrally symmetric.
\end{thm}

\begin{proof}
Let $S:\left[0,1\right]\to\mathcal{K}$ denote a section. Let $\pi_{0}\in\mathcal{K}$
denote an arbitrary extreme point and let $\pi_{1}$ denote a point
on the boundary such that $\left(1-s\right)\cdot\pi_{0}+s\cdot\pi_{1}=S\left(\boldsymbol{\nicefrac{1}{2}}\right)$
where $0\leq s\leq\nicefrac{1}{2}.$ Then $S\left(\nicefrac{1}{2}\right)$ can be written as a mixture $\left(1-t\right)\cdot\sigma_{0}+t\cdot\sigma_{1}$ 
of points on the boundary such that $t$ is minimal. As in the proof of Theorem \ref{thm:SectionGennemPunkt} we see that $\sigma_{0}$ and $\sigma_{1}$ 
are orthogonal.  Since $\mathcal{K}$ is spectral we have $t=\nicefrac{1}{2}$. Since $t\leq s\leq \nicefrac{1}{2}$ we have $s=\nicefrac{1}{2}$ implying that $\mathcal{K}$
is symmetric around $S\left(\nicefrac{1}{2}\right).$ 
\end{proof}
\begin{prop}\label{prop:fix}
Let $S:\mathcal{L}\to\mathcal{K}$ denote a section with retraction
$R:\mathcal{K}\to\mathcal{L}.$ If $\mathcal{K}$ is a spectral set
of rank 2 and $\mathcal{L}$ is not a singleton then $\mathcal{L}$ is also a spectral set of rank 2. If
$c$ is the center of $\mathcal{K}$ then $R\left(c\right)$ is the
center of $\mathcal{L}$ and $S\left(R\left(c\right)\right)=c,$ i.e.
the section goes through the center of $\mathcal{K}.$
\end{prop}

\begin{proof}
Let $\sigma\to\tilde{\sigma}$ denote reflection in the point $c\in\mathcal{K}.$
If $\rho\in\mathcal{L}$ then 
\begin{align}
R\left(c\right) & =R\left(\frac{1}{2}\cdot S\left(\rho\right)+\frac{1}{2}\cdot\widetilde{S\left(\rho\right)}\right)\\
 & =\frac{1}{2}\cdot R\left(S\left(\rho\right)\right)+\frac{1}{2}\cdot R\left(\widetilde{S\left(\rho\right)}\right)\\
 & =\frac{1}{2}\cdot\rho+\frac{1}{2}\cdot R\left(\widetilde{S\left(\rho\right)}\right)
\end{align}
so that $\mathcal{L}$ is centrally symmetric around $R\left(c\right)$.
If $\mathcal{F}$ is a proper face of $\mathcal{L}$ then $S\left(\mathcal{F}\right)$
is a proper face of $\mathcal{K}$ implying that $S\left(\mathcal{F}\right)$
is a singleton. Therefore $\mathcal{F}=R\left(S\left(\mathcal{F}\right)\right)$
is a singleton implying that $\mathcal{L}$ has rank 2. If $\rho\in\mathcal{L}$
is an extreme point then $\tilde{\rho}=R\left(\widetilde{S\left(\rho\right)}\right)$
is also an extreme point of $\mathcal{L}$. Now 
\begin{align}
S\left(R\left(c\right)\right) & =S\left(\frac{1}{2}\cdot\rho+\frac{1}{2}\cdot R\left(\widetilde{S\left(\rho\right)}\right)\right)\\
 & =\frac{1}{2}\cdot S\left(\rho\right)+\frac{1}{2}\cdot S\left(R\left(\widetilde{S\left(\rho\right)}\right)\right).
\end{align}
Since $\tilde{\rho}\in\mathcal{L}$ is an extreme point and $\tilde{\rho}=R\left(\widetilde{S\left(\rho\right)}\right)$
we have that $R^{-1}\left(\tilde{\rho}\right)$ is a proper face of
$\mathcal{K}$ and thereby a singleton. Therefore $S\left(R\left(\widetilde{S\left(\rho\right)}\right)\right)=\widetilde{S\left(\rho\right)}$
and $S\left(R\left(c\right)\right)=\frac{1}{2}\cdot S\left(\rho\right)+\frac{1}{2}\cdot\widetilde{S\left(\rho\right)}=c.$
\end{proof}
\begin{cor}
If $\sigma$ is an extreme point of a spectral set $\mathcal{K}$
of rank 2 then there exists a unique element in $\mathcal{K}$ that
is orthogonal to $\sigma\, .$ 
\end{cor}

If a centrally symmetric set has a proper face that is not an extreme
point then the set is not spectral as illustrated in Figure \ref{fig:CentrallySymmetricWithFace}.

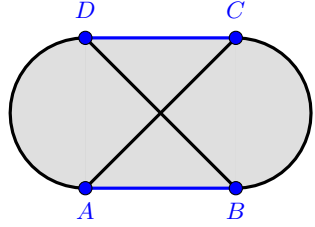
\begin{figure}[tbh]
\begin{centering}
\begin{tikzpicture}[line cap=round,line join=round,>=triangle 45,x=1.0cm,y=1.0cm]
\clip(1.82,0.54) rectangle (6.28,3.54);
\fill[line width=0pt,color=gray!30!white,fill=gray!50!white,fill opacity=0.5] (3.,3.) -- (5.,3.) -- (5.,1.) -- (3.,1.) -- cycle; 
\draw [shift={(5.,2.)},line width=1.2pt,color=black,fill=gray!50!white,fill opacity=0.5]  plot[domain=-1.5707963267948966:1.5707963267948966,variable=\t]({1.*1.*cos(\t r)+0.*1.*sin(\t r)},{0.*1.*cos(\t r)+1.*1.*sin(\t r)});
\draw [shift={(3.,2.)},line width=1.2pt,color=black,fill=gray!50!white,fill opacity=0.5]  plot[domain=1.5707963267948966:4.71238898038469,variable=\t]({1.*1.*cos(\t r)+0.*1.*sin(\t r)},{0.*1.*cos(\t r)+1.*1.*sin(\t r)});

\draw [line width=1.2pt] (3.,3.)-- (5.,1.);
\draw [line width=1.2pt] (3.,1.)-- (5.,3.);

\draw [line width=1.2pt, color=blue] (3.,1.)-- (5.,1.);
\draw [fill=blue] (3.,1.) circle (2.5pt);
\draw[color=blue] (3.0,0.7) node {$A$};
\draw [fill=blue] (5.,1.) circle (2.5pt);
\draw[color=blue] (5.0,0.7) node {$B$};

\draw [line width=1.2pt, color=blue] (5.,3.)-- (3.,3.);
\draw [fill=blue] (5.,3.) circle (2.5pt);
\draw[color=blue] (5.0,3.37) node {$C$};
\draw [fill=blue] (3.,3.) circle (2.5pt);
\draw[color=blue] (3.0,3.37) node {$D$};

\end{tikzpicture}
\par\end{centering}
\caption{\label{fig:CentrallySymmetricWithFace}A centrally symmetric convex
body with non-trivial faces $\overline{AB}$ and $\overline{CD}$. The Caratheodory number is 2, but the rank is 3.
The points $A,B,C$, and $D$ are orthogonal extreme points and any
point in the interior of the square $\square ABCD$ has several orthogonal
decompositions with different mixing coefficients with weights on
$A,B,C$, and $D$. Points in the convex body but outside the triangles
can be decomposed as a mixture of two orthogonal extreme points on
the semi circles. }
\end{figure}

Let $\rho=x\cdot\sigma_{0}+y\cdot\sigma_{1}$ denote an orthogonal
decomposition of an element of the vector space generated by a spectral
set of rank 2. Then we may define 
\begin{equation}\label{eq:spectral_theory}
f\left(\rho\right)=f\left(x\right)\cdot\sigma_{0}+f\left(y\right)\cdot\sigma_{1}\, .
\end{equation}
 If $\rho=x\cdot\rho_{0}+y\cdot\rho_{1}$ is another orthogonal decomposition
then $x=y$ and 
\begin{equation}
f\left(x\right)\cdot\sigma_{0}+f\left(y\right)\cdot\sigma_{1}=2f\left(x\right)\cdot\frac{\sigma_{0}+\sigma_{1}}{2}
\end{equation}
 and 
\begin{equation}
f\left(x\right)\cdot\rho_{0}+f\left(y\right)\cdot\rho_{1}=2f\left(x\right)\cdot\frac{\rho_{0}+\rho_{1}}{2}.
\end{equation}
Since 
\begin{equation}
\frac{\sigma_{0}+\sigma_{1}}{2}=\frac{\rho_{0}+\rho_{1}}{2}=c
\end{equation}
different orthogonal decompositions will result in the same value
of $f\left(\rho\right).$ Note in particular that for the constant
function $f\left(x\right)=\nicefrac{1}{2}$ we have $f\left(\rho\right)=c.$
In this sense $c=\nicefrac{1}{2}$ and from now on we will use $\boldsymbol{\frac{1}{2}}$ in bold face instead of $c$  as notation for the center of a spectral set. If $\frac{1}{2}\cdot\rho+\frac{1}{2}\cdot\sigma=c$
then $\frac{1}{2}\cdot\rho+\frac{1}{2}\cdot\sigma=\boldsymbol{\frac{1}{2}}$ so
that $\rho+\sigma=\boldsymbol{1}$ so that the central inversion of $\rho$ equals
$\boldsymbol{1}-\rho.$ We note that if $f\left(x\right)\geq0$ for all $x$ then
$f\left(\rho\right)$ is element in the positive cone. Therefore $\sum_{i}\rho_{i}^{2}=0$
implies that $\rho_{i}=0$ for all $i$, where $\rho_{i}^{2}$ is defined via Equation (\ref{eq:spectral_theory}). Note also that if $\Phi$
is an isomorphism then $\Phi\left(f\left(\rho\right)\right)=f\left(\Phi\left(\rho\right)\right).$

\section{Sufficient Regret Functions}

There are a number of equivalent ways of defining sufficiency, and
the present definition of sufficiency is based on \cite{Petz1988}.
We refer to \cite{Jencova2006} where the notion of sufficiency is
discussed in great detail.
\begin{defn}
Let $\left(\sigma_{\theta}\right)_{\theta}$ denote a family of points
in a convex body $\mathcal{K}$ and let $\Phi$ denote an affinity
$\mathcal{K}\to\mathcal{L}$ where $\mathcal{K}$ and $\mathcal{L}$
denote convex bodies. Then $\Phi$ is said to be \emph{sufficient}
for $\left(\sigma_{\theta}\right)_{\theta}$ if there exists an affinity
$\Psi:\mathcal{L}\to\mathcal{K}$ such that $\Psi\left(\Phi\left(\sigma_{\theta}\right)\right)=\sigma_{\theta},$
i.e. the states $\sigma_{\theta}$ are fix-points of $\Psi\circ\Phi\, .$ 
\end{defn}

The notion of sufficiency as a property of general divergences was
introduced in \cite{Harremoes2007a}. It was shown in \cite{Jiao2014}
that a Bregman divergence on the simplex of distributions on an alphabet
that is not binary determines the divergence up to a multiplicative
factor. In \cite{Harremoes2017} this result was extended to $C^{^{*}}$-algebras.
Here we are interested in the binary case and its generalization that
is convex bodies of rank 2.
\begin{defn}
We say that the regret function $D_{F}$ on the convex body $\mathcal{K}$
satisfies \emph{sufficiency} if 
\begin{equation}
D_{F}\left(\Phi\left(\rho\right),\Phi\left(\sigma\right)\right)=D_{F}\left(\rho,\sigma\right)
\end{equation}
for any affinity $\mathcal{K}\to\mathcal{K}$ that is sufficient for
$\left(\rho,\sigma\right).$ 
\end{defn}

\begin{lem}
If a strict regret function on a convex body of rank 2 satisfies sufficiency,
then the convex body is spectral and the regret function is generated
by a function of the form 
\begin{equation}
F\left(\sigma\right)=\mathrm{tr}\left[f\left(\sigma\right)\right]\label{eq:trace}
\end{equation}
for some convex function $f:\left[0,1\right]\to\mathbb{R}.$ 
\end{lem}

\begin{proof}
For $i=1,2$ assume that $S_{i}:\left[0,1\right]\to\mathcal{K}$ are
sections with retractions $R_{i}:\mathcal{K}\to\left[0,1\right].$
Then $S_{2}\circ R_{1}$ is sufficient for the pair $\left(S_{1}\left(t\right),S_{1}\left(\nicefrac{1}{2}\right)\right)$
with recovery map $S_{1}\circ R_{2}$ implying that 
\begin{equation}
D_{F}\left(S_{1}\left(t\right),S_{1}\left(\nicefrac{1}{2}\right)\right)=D_{F}\left(S_{2}\left(t\right),S_{2}\left(\nicefrac{1}{2}\right)\right).
\end{equation}
Define $f\left(t\right)=D_{F}\left(S_{1}\left(t\right),S_{1}\left(\nicefrac{1}{2}\right)\right).$
Then $D_{F}\left(S_{2}\left(t\right),S_{2}\left(\nicefrac{1}{2}\right)\right)=f\left(t\right)$
for any section $S_{2}$, so this divergence is completely determined
by the spectrum $\left(t,1-t\right).$ In particular all orthogonal
decompositions have the same spectrum so that the convex body
is spectral.

Let $\mathcal{K}$ denote a spectral convex set of rank 2 with center
$\boldsymbol{\frac{1}{2}}\,.$ If the Bregman divergence $D_{F}$ satisfies sufficiency then
$D_{F}\left(\rho,\sigma\right)=D_{F}\left(\boldsymbol{1}-\rho,\boldsymbol{1}-\sigma\right)$
and 
\begin{align}
D_{F}\left(\rho,\sigma\right) & =\frac{D_{F}\left(\rho,\sigma\right)+D_{F}\left(\boldsymbol{1}-\rho,\boldsymbol{1}-\sigma\right)}{2}\\
 & =\frac{D_{F}\left(\rho,\sigma\right)+D_{\tilde{F}}\left(\rho,\sigma\right)}{2}\\
 & =D_{\frac{F+\tilde{F}}{2}}\left(\rho,\sigma\right)
\end{align}
where $\tilde{F}\left(\sigma\right)$ is defined as $F\left(\boldsymbol{1}-\sigma\right).$
Now $\frac{F+\tilde{F}}{2}$ is convex and invariant under central
inversion. Therefore a regret function on a spectral set of rank 2
is generated by a function that is invariant under central inversion. 

Let $F$ denote a convex function that is invariant under central
inversion and assume that $D_{F}$ satisfies sufficiency. If $\sigma_{0}$ and $\sigma_{1}$ are orthogonal we may define $f\left(t\right)=\frac{1}{2}\cdot F\left(\left(1-t\right)\cdot\sigma_{0}+t\cdot\sigma_{1}\right)$
for $t\in\left[0,1\right]$. Then 
\begin{align}
\mathrm{tr}\left[ f\left(\sigma \right)\right] & = \mathrm{tr}\left[ f\left(1-t\right)\cdot\sigma_{0}+f\left(t \right)\cdot\sigma_{1}\right]\\
&= f\left(1-t\right)\cdot1+f\left(t \right)\cdot1\\
&=2\cdot f\left(t \right)\\
&=2\cdot\frac{1}{2}\cdot F\left(\left(1-t\right)\cdot\sigma_{0}+\left(t \right)\cdot\sigma_{1} \right)\\
&=F\left(\sigma\right)\, ,
\end{align}
which proves Eq. (\ref{eq:trace}).
\end{proof}

\begin{prop}
Let $\mathcal{K}$ denote a spectral convex set of rank 2. If $f:\left[0,1\right]\to\mathbb{R}$
is convex then $F\left(\sigma\right)=\mathrm{tr}\left[f\left(\sigma\right)\right]$
defines a convex function on $\mathcal{K}$ and the regret function
$D_{F}$ satisfies sufficiency.
\end{prop}

\begin{proof}
Let $\rho_{0}$ and $\rho_{1}$ denote points in $\mathcal{K}$. Let
$\sigma$ denote a point that is co-linear with $\boldsymbol{\frac{1}{2}}$ and $\rho_{0}$
and such that $F\left(\sigma\right)=F\left(\rho_{1}\right).$ Then
\begin{align}
F\left(\left(1-t\right)\cdot\rho_{0}+t\cdot\rho_{1}\right) & \leq F\left(\left(1-t\right)\cdot\rho_{0}+t\cdot\sigma\right)\\
 & =\mathrm{tr}\left[f\left(\left(1-t\right)\cdot\rho_{0}+t\cdot\sigma\right)\right]\\
 & \leq\mathrm{tr}\left[\left(1-t\right)\cdot f\left(\rho_{0}\right)+t\cdot f\left(\sigma\right)\right]\\
 & =\left(1-t\right)\cdot F\left(\rho_{0}\right)+t\cdot F\left(\sigma\right)\\
 & =\left(1-t\right)\cdot F\left(\rho_{0}\right)+t\cdot F\left(\rho_{1}\right)\, ,
\end{align}
which proves that $F$ is convex.

Now we will prove that $D_F$ satisfies sufficiency. Let $\rho , \sigma\in\mathcal{K}$ denote two point and let $\Phi :\mathcal{K}\to\mathcal{K}$  denote an affinity that is sufficient for $\rho , \sigma$ with recovery map $\Psi$. 
Then $\Phi\circ\Psi$ and $\Psi\circ\Phi$ are retractions and the fixpoint set of $\Psi\circ\Phi$ and $\Phi\circ\Psi$ are isomorphic convex bodies. 
Accoring to Proposition \ref{prop:fix} the center of $\mathcal{K}$ a fixpoint under retractions and we see that a decomposition into 
orthogonal extreme point in a fixpoint set is also an orthogonal decomposition in $\mathcal{K}$. 
Therefore $\mathrm{tr}\left[f\left(\sigma\right)\right]$ has the same value when the calculation is done within the fixpoint set of $\Psi\circ\Phi$, which proves the proposition.
\end{proof}
\begin{thm}
Let $\mathcal{K}$ denote a convex body of rank 2 with a sufficient
Bregman divergence $D_{F}$ that is strict. Then the center of $\mathcal{K}$ the unique point that achieves
the minimax regret.
\end{thm}

\begin{proof}
Let $S:\left[0,1\right]\to\mathcal{K}$ denote a section. Then $S\left(\nicefrac{1}{2}\right)=\boldsymbol{\frac{1}{2}}$ and
\begin{align}
C_{F} & \geq\frac{1}{2}\cdot D_{F}\left(S\left(0\right),S\left(\nicefrac{1}{2}\right)\right)+\frac{1}{2}\cdot D_{F}\left(S\left(1\right),S\left(\nicefrac{1}{2}\right)\right)+D_{F}\left(S\left(\nicefrac{1}{2}\right),\sigma_{opt}\right)\\
 & =D_{F}\left(S\left(1\right),\boldsymbol{\frac{1}{2}}\right)+D_{F}\left(\boldsymbol{\frac{1}{2}},\sigma_{opt}\right)
\end{align}
 Further we have 
\begin{equation}
\sup_{\rho\in\mathcal{K}}D_{F}\left(\rho,\boldsymbol{\frac{1}{2}}\right)\geq C_{F}+D_{F}\left(\sigma_{opt},\boldsymbol{\frac{1}{2}}\right).
\end{equation}
Now $\rho=S_{\rho}\left(t\right)$ for some section $S_{\rho}$ and
some $t\in\left[0,1\right].$ Therefore 
\begin{align}
D_{F}\left(\rho,\boldsymbol{\frac{1}{2}}\right) & =D_{F}\left(S_{\rho}\left(t\right),S\left(\nicefrac{1}{2}\right)\right)\\
 & =D_{F}\left(S_{\rho}\left(t\right),S_{\rho}\left(\nicefrac{1}{2}\right)\right)\\
 & =D_{F}\left(S\left(t\right),S\left(\nicefrac{1}{2}\right)\right)\\
 & \leq D_{F}\left(S\left(1\right),\boldsymbol{\frac{1}{2}}\right).
\end{align}
Therefore $C_{F}=D_{F}\left(S\left(1\right),\boldsymbol{\frac{1}{2}}\right)$
and $D_{F}\left(\sigma_{opt},\boldsymbol{\frac{1}{2}}\right)=0$
implying $\sigma_{opt}=\boldsymbol{\frac{1}{2}}\, .$
\end{proof}
If the Bregman divergence is based on Shannon entropy then the minimax
regret is called the capacity and the result is that a convex body
of rank 2 has a capacity of 1 bit.

\section{Spin Factors}

We say that a convex body is a \emph{Hilbert ball} if the convex body
can be embedded as a unit ball in a $d$ dimensional real Hilbert
space $\mathcal{H}$ with some inner product that will be denoted
$\left\langle \cdot\mid\cdot\right\rangle .$ The positive elements are
the elements $\left(\vec{v},s\right)$ where $\left\Vert \vec{v}\right\Vert _{2}\leq s.$
The trace
of the spin factor is $\mathrm{tr}\left[\left(\vec{v},s\right)\right]=2s.$ 

The direct sum $\mathcal{H}\oplus\mathbb{R}$
can be equipped a product $\bullet$ by 
\begin{equation}
\left(\vec{v},s\right)\bullet\left(\vec{w},t\right)=\left(t\cdot\vec{v}+s\cdot\vec{w},\left\langle \vec{v}\left|\vec{w}\right.\right\rangle +s\cdot t\right)\, .
\end{equation}
This product is distributive and $\left(\vec{v},1\right)\bullet\left(-\vec{v},1\right)=0.$
Therefore $x^{2}$ defined via (\ref{eq:spectral_theory}) will be equal to $x\bullet x$ and $\left(\mathcal{H}\oplus\mathbb{R},\bullet\right)$
becomes a \emph{formally real Jordan algebra} of the type that is called a
\emph{spin factor} and is denoted $JSpin_{d}$. The unit of a spin factor is $\left(\vec{0},1\right)$ and will be denoted $\boldsymbol{1}.$ See \cite{McCrimmon2004}
for general results on Jordan algebras. 

Let $\mathcal{M}_{n}\left(\mathbb{F}\right)$ denote $n\times n$ matrices over
$\mathbb{F}$ where $\mathbb{F}$ may denote the real numbers $\mathbb{R}$
or the complex numbers $\mathbb{C}$ or the quaternions $\mathbb{H}$
or the octonions $\mathbb{O}.$ Let $\left(\mathcal{M}_{n}\left(\mathbb{F}\right)\right)_{h}$
denote the set of self-adjoint matrices of $\mathcal{M}_{n}\left(\mathbb{F}\right).$
Then $\left(\mathcal{M}_{n}\left(\mathbb{F}\right)\right)_{h}$ is
a formally real Jordan algebra with a Jordan product $\bullet$ is
given by 
\begin{equation}
x\bullet y=\frac{1}{2}\left(xy+yx\right)
\end{equation}
except for $\mathbb{F}=\mathbb{O}$ where one only get a Jordan algebra
when $n\leq3.$ The self-adjoint $2\times2$ matrices with real, complex,
quaternionic or octonionic entries can be identified with spin factors
with dimension $d=2,$ $d=3,$ $d=5,$ or $d=9.$ The most important
examples of spin factors are the bit $JSpin_{1}$ and the qubit $JSpin_{3}$. 

We introduce the Pauli matrices 
\begin{equation}
\boldsymbol{\sigma}_{1}=\left(\begin{array}{cc}
0 & 1\\
1 & 0
\end{array}\right),\,\,\boldsymbol{\sigma}_{3}=\left(\begin{array}{cc}
1 & 0\\
0 & -1
\end{array}\right)
\end{equation}
 and observe that $\boldsymbol{\sigma}_{1}\bullet\boldsymbol{\sigma}_{3}=\boldsymbol{0}.$
Let $\boldsymbol{v}_{1},\boldsymbol{v}_{2},\dots,\boldsymbol{v}_{d}$ denote a basis of
the Hilbert space $\mathcal{H}.$ Let the function $S:JSpin_{d}\to\left(\mathcal{M}_{2}\left(\mathbb{R}\right)\right)^{\otimes \left(d-1\right)}$
be defined by
\begin{align}
S\left(\boldsymbol{1}\right) & =\boldsymbol{1}\otimes\boldsymbol{1}\otimes\boldsymbol{1}\otimes\dots\otimes\boldsymbol{1}\, ,\\
S\left(\boldsymbol{v}_{1}\right) & =\boldsymbol{\sigma}_{1}\otimes\boldsymbol{1}\otimes\boldsymbol{1}\otimes\dots\otimes\boldsymbol{1}\, ,\\
S\left(\boldsymbol{v}_{2}\right) & =\boldsymbol{\sigma}_{3}\otimes\boldsymbol{\sigma}_{1}\otimes\boldsymbol{1}\otimes\dots\otimes\boldsymbol{1}\, ,\\
S\left(\boldsymbol{v}_{3}\right) & =\boldsymbol{\sigma}_{3}\otimes\boldsymbol{\sigma}_{3}\otimes\boldsymbol{\sigma}_{1}\otimes\dots\otimes\boldsymbol{1}\, ,\\
&\,\,\,\vdots\\
S\left(\boldsymbol{v}_{d-1}\right) & =\boldsymbol{\sigma}_{3}\otimes\boldsymbol{\sigma}_{3}\otimes\boldsymbol{\sigma}_{3}\otimes\dots\otimes\boldsymbol{\sigma}_{1}\, ,\\
S\left(\boldsymbol{v}_{d}\right) & =\boldsymbol{\sigma}_{3}\otimes\boldsymbol{\sigma}_{3}\otimes\boldsymbol{\sigma}_{3}\otimes\dots\otimes\boldsymbol{\sigma}_{3}\, .
\end{align}
Then $S$ can be linearly extended and one easily checks that 
\begin{equation}
S\left(x\bullet y\right)=S\left(x\right)\bullet S\left(y\right).
\end{equation}
Now $S\left(JSpin_{d}\right)$ is a linear subspace of the real Hilbert
space $\left(\mathcal{M}_{2}\left(\mathbb{R}\right)\right)^{\otimes \left(d-1\right)}$
so there exists a projection of $\left(\mathcal{M}_{2}\left(\mathbb{R}\right)\right)^{\otimes \left(d-1\right)}$
onto $S\left(JSpin_{d}\right)$ and this projection maps symmetric
matrices in $\left(\mathcal{M}_{2}\left(\mathbb{R}\right)\right)^{\otimes \left(d-1\right)}$
into symmetric matrices. Therefore $S$ is a section with a retraction
generated by the projection. In this way $JSpin_{d}$ is a section
of a Jordan algebra of symmetric matrices with real entries. The Jordan
algebra $\mathcal{M}_{n}\left(\mathbb{R}\right)_{h}$ is obviously
a section of $\mathcal{M}_{n}\left(\mathbb{C}\right)_{h}$ so $JSpin_{d}$ is
a section of $\mathcal{M}_{n}\left(\mathbb{C}\right)_{h}.$ Note that the projection of $\mathcal{M}_{n}\left(\mathbb{C}\right)_{h}$ on a spin factor is not necessarily completely positive. 

Since the standard
formalism of quantum theory represents states as density matrices in
$\mathcal{M}_{n}\left(\mathbb{C}\right)_{h}$ we see that spin factors appear
as sections of state spaces of the usual formalism of quantum theory.
Therefore the points in the Hilbert ball are called \emph{states}
and the Hilbert ball is called the \emph{state space} of the spin factor.
The extreme points in the state space are called \emph{pure states}. 

The positive cone of a spin factor is self-dual in the sense that
any positive functional $\phi:JSpin_{d}\to\mathbb{R}$ is given by
$\phi\left(x\right)=\mathrm{tr}\left[x\bullet y\right]$ for some
uniquely determined positive element $y.$ We recall the definition of the polar set of a convex body $\mathcal{K}\subseteq \mathbb{R}^d$
\begin{equation}
\mathcal{K}^{\circ} =\left\{ y\in\mathbb{R}^d\mid \langle x,y\rangle\leq1\text{ for all }x\in\mathcal{K}\right\}\, .
\end{equation}
\begin{prop}
Assume that the cone generated by a spectral convex body $\mathcal{K}$ of rank 2
is self-dual. Then it can be represented as a spin factor. 
\end{prop}

\begin{proof}
If $\phi$ is a test on $\mathcal{K}$ then $2\cdot\phi -1$ maps $\mathcal{K}$ into $[0,1]$, 
which an element in the polar set of $\mathcal{K}$ embedded in a Hilbert space with the center as the origin. 
Since the cone is assumed to be self-dual the set $\mathcal{K}$ is self-polar and Hilbert balls are the the only self-polar sets. 
The result follows because a Hilbert ball can be represented as the state space of a spin factor.
\end{proof}
A convex body $\mathcal{K}$ of rank 2 is said to have \emph{symmetric
transission probabilities} if for any extreme points $\sigma_{1}$
and $\sigma_{2}$ there exists retractions $R_{1}:\mathcal{K}\to\left[-1,1\right]$
and $R_{2}:\mathcal{K}\to\left[-1,1\right]$ such that $R_{i}\left(\sigma_{i}\right)=1$,
and $R_{1}\left(\sigma_{2}\right)=R_{2}\left(\sigma_{1}\right).$

\begin{thm}
A spectral convex body $\mathcal{K}$ of rank 2 with symmetric transmission
probabilities can be represented by a spin factor.
\end{thm}

\begin{proof}
 For almost all extreme points $\sigma$ of $\mathcal{K}$ a retraction $R:\mathcal{K}\to\left[-1,1\right]$ with $R\left(\sigma\right)=1$ is uniquely determined. 
Let $\sigma_1$ and $\sigma_2$ be two extreme points that are not antipodal and with unique retractions $R_1$ and $R_2$. Let $\mathcal{L}$ denote the intersection of $\mathcal{K}$ with the affine span of $\sigma_1 ,\, \sigma_2$ and the center.
Embed $\mathcal{L}$ in a 2-dimensional coordinate system with the center of $\mathcal{L}$ as origin of the coordinate system. 
Let $\sigma$ denote an extreme point with a unique retraction $R$. Then $R_{1}\left(\sigma\right)\cdot \sigma-\sigma_{1}$ is parallel with $R_{2}\left(\sigma\right)\cdot \sigma-\sigma_{2}$ because
\begin{align}
R\left(R_{i}\left(\sigma\right)\cdot \sigma-\sigma_{i}\right)&=R_{i}\left(\sigma\right)\cdot R\left(\sigma\right)-R\left(\sigma_{i}\right)\\
&=R_{i}\left(\sigma\right)\cdot 1-R_{i}\left(\sigma\right)\\
&=0.
\end{align}
Therefore the determinant of $R_{1}\left(\sigma\right)\cdot \sigma-\sigma_{1}$ and $R_{2}\left(\sigma\right)\cdot \sigma-\sigma_{2}$ is zero, but the determinant can be calculated as
\begin{multline}
\det\left(R_{1}\left(\sigma\right)\cdot \sigma-\sigma_{1},R_{2}\left(\sigma\right)\cdot \sigma-\sigma_{2}\right)\\
=0-\det\left( R_{1}\left(\sigma\right)\cdot \sigma,\sigma_{2}\right)
-\det\left(\sigma_{1},R_{2}\left(\sigma\right)\cdot \sigma\right)+\det\left(\sigma_{1},\sigma_{2}\right)\\
=\det\left(  \sigma,R_{2}\left(\sigma\right)\cdot\sigma_{1}-R_{1}\left(\sigma\right)\cdot\sigma_{2}\right)
-\det\left(\sigma_{2},\sigma_{1}\right)\, .
\end{multline}
This means that $\sigma$ satisfies the following equation
\begin{equation}
\det\left(  \sigma,R_{2}\left(\sigma\right)\cdot\sigma_{1}-R_{1}\left(\sigma\right)\cdot\sigma_{2}\right)
=\det\left(\sigma_{2},\sigma_{1}\right) .
\end{equation}
This is a quadratic equation in the coordinates of $\sigma$, which implies that $\sigma$ lies on a conic section. Since $\mathcal{L}$ is bounded this conic section must be a circle or an ellipsoid.
Almost
all extreme points of $\mathcal{L}$ have unique retractions.
Therefore almost all extreme points lie on a circle or an ellipsoid which by convexity implies that all extreme points of $\mathcal{L}$ lie on an ellipsoid or a circle. 
Since this holds for almost all pairs $\sigma_1$ and $\sigma_2$ the convex set $\mathcal{K}$ must be an ellipsoid, which can be mapped into a ball.
\end{proof}
\begin{defn}
Let $\mathcal{A}\subseteq JSpin_{d}$ denote a subalgebra of a spin
factor. Then $\mathbb{E}:JSpin_{d}\to A$ is called a conditional
expectation if $\mathbb{E}\left(1\right)=1$ and $\mathbb{E}\left(a\bullet x\right)=a\bullet\mathbb{E}\left(x\right)$ for any $a\in\mathcal{A}$.
\end{defn}

\begin{thm}
Let $\mathcal{K}$ denote the state space of a spin factor and assume that
$\Phi:\mathcal{K}\to\mathcal{K}$ is an idempotent that preserves the center. Then $\Phi$
is a conditional expectation of the spin factor into a sub-algebra
of the spin factor. 
\end{thm}

\begin{proof}
Assume that the spin factor is based on the Hilbert space $\mathcal{H}=\mathcal{H}_{1}\oplus\mathcal{H}_{2}$
and that the idempotent $\Phi$ is the identity on $\mathcal{H}_{1}$
and maps $\mathcal{H}_{2}$ into the origin. Let $\boldsymbol{v},\boldsymbol{w}_{1}\in H_{1}$
and $\boldsymbol{w}_{2}\in H_{2}$ and $s,t\in R.$ Then 
\begin{align}
\Phi\left(\left(\boldsymbol{v},s\right)\bullet\left(\boldsymbol{w}_{1}+\boldsymbol{w}_{2},t\right)\right) 
& =\Phi\left(\left(\boldsymbol{v},s\right)\bullet\left(\boldsymbol{w}_{1},t\right)\right)+\Phi\left(\left(\boldsymbol{v},s\right)\bullet\left(\boldsymbol{w}_{2},0\right)\right)\\
 & =\left(\boldsymbol{v},s\right)\bullet\left(\boldsymbol{w}_{1},t\right)+\Phi\left(s\cdot\boldsymbol{w}_{2},\left\langle \boldsymbol{v},\boldsymbol{w}_{2}\right\rangle \right)\\
 & =\left(\boldsymbol{v},s\right)\bullet\Phi\left(\boldsymbol{w}_{1}+\boldsymbol{w}_{2},t\right)\, ,
\end{align}
which proves the theorem.
\end{proof}

\section{Monotonicity under dilations}

Next we introduce the notion of monotonicity. In thermodynamics monotonicity is associated with decrease of free
energy in a closed system and in information theory it is associated with the data processing inequality.
\begin{defn}
Let $D_{F}$ denote a regret function on the convex body $\mathcal{K}.$
Then $D_{F}$ is said to be \emph{monotone} if 
\begin{equation}
D_{F}\left(\Phi\left(\rho\right),\Phi\left(\sigma\right)\right)\leq D_{F}\left(\rho,\sigma\right)
\end{equation}
for any affinity $\Phi:\mathcal{K}\to\mathcal{K}.$
\end{defn}

A simple example of a monotone regret function is squared Euclidean distance in a Hilbert ball, but later we shall see that there are many other examples. All monotone regret functions are Bregman
divergences \cite[Prop. 6]{Harremoes2017} that satisfy sufficiency
\cite[Prop. 8]{Harremoes2017}. We shall demonstrate that a convex body of rank 2 with a monotone Bregman divergence can be represented by a spin factor.

We will need to express the Bregman divergence as an integral
involving a different type of divergence. Define
\begin{equation}
D^{F}\left(x,y\right)=\frac{\mathrm{d}^{2}}{\mathrm{d}s^{2}}F\left(x_{s}\right)_{\mid s=1}
\label{eq:AndenAfledt}
\end{equation}
where $x_{s}=\left(1-s\right)\cdot x+s\cdot y.$ If $F$ is $C^2$ and $\mathbf{H}\left(y\right)$ is the
Hesse matrix of $F$ calculated in the point $y$ then 
\begin{equation}
D^{F}\left(x,y\right)=\left\langle x-y \left| \mathbf{H}\left(y\right)\right| x-y\right\rangle .\label{eq:Hesse}
\end{equation}
Since 
\begin{equation}
F\left(x_{s}\right)=F\left(y\right)+\left< \nabla F\left(y\right)\mid x_{s}-y \right>+D_{F}\left(x_{s},y\right)
\end{equation}
we also have 
\begin{equation}\label{eq:derivdiv}
D^{F}\left(x,y\right)=\frac{\mathrm{d}^{2}}{\mathrm{d}s^{2}}D_{F}\left(x_{s},y\right)_{\mid s=1}.
\end{equation}
It is also easy to verify that
\begin{equation}\label{eq:deriv2div}
D^{F}\left(x,y\right)=\frac{\mathrm{d}}{\mathrm{d}s}D_{F}\left(y,x_{s}\right)_{\mid s=1}.
\end{equation}

\begin{prop}\label{prop:lokal}
Let $F:\left[0,1\right]\to\mathbb{R}$ denote a twice differentiable
convex function. If $x_{s}=\left(1-s\right)\cdot x+s\cdot y.$ Then 
\begin{equation}
D_{F}\left(x,y\right)=\int_{0}^{1}\frac{D^{F}\left(x,x_{s}\right)}{s}\,\mathrm{d}s\,.
\end{equation}
where $D^{F}\left(x,x_{s}\right)$ is given by one of the equations (\ref{eq:AndenAfledt}), (\ref{eq:Hesse}), (\ref{eq:derivdiv}), or (\ref{eq:deriv2div}).
\end{prop}
A similar result appear in \cite{Hayashi2016} as Eq. 2.118. In the context of complex matrices the result was proved as Proposition \ref{prop:lokal} in \cite{Pitrik2015}.

We will need the following lemma.
\begin{lem}\label{lem:andenorden}
Let $F$ denote a convex function defined on a convex body $\mathcal{K}$. Then for almost all $y\in\mathcal{K}$ we have
\begin{equation}
\lim_{x\to y}\frac{D_{F}\left(x,y\right)-\frac{1}{2}D^{F}\left(x,y\right)}{\| x-y\|^2}=0\, .
\end{equation}
\end{lem}
\begin{proof}
According to our definitions
\begin{multline}
D_{F}\left(x,y\right)-\frac{1}{2}D^{F}\left(x,y\right)\\
= F\left(x\right)-\left(F\left(y\right)+ \langle \nabla F\left(y\right)\mid x-y\rangle+\frac{1}{2}\langle x-y \left|\mathbf{H}\left(y\right)\right| x-y\rangle\right)
\end{multline}
and we see that Lemma \ref{lem:andenorden} states that a convex function is twice differentiable almost everywhere, which is exactly Alexandrov's theorem \cite{Alexandrov1939}.
\end{proof}

\begin{lem}\label{lem:lokal}
If $F$ is twice differentiable then $D_F$ is a monotone Bregman divergence if and only if $D^F$ is monotone.
\end{lem}
\begin{proof}
Assume that $D_F$ is monotone and that $\Phi$ is some affinity and that $x_{s} =\left(1-s\right)\cdot x+ s\cdot y$. Then
\begin{equation}
D_{F}\left(\Phi\left(x_{s}\right),\Phi\left(y\right)\right)\leq D_{F}\left(x_{s} ,y\right) .
\end{equation}
Since 
\begin{align}
D_{F}\left(x_{s} ,y\right) &=0\\
\frac{\mathrm{d}}{\mathrm{d}s}D_{F}\left(x_{s},y\right)_{\mid s=1}&=0
\end{align}
and
\begin{align}
D_{F}\left(\Phi\left(x_{s}\right) ,\Phi\left(y\right)\right) &=0\\
\frac{\mathrm{d}}{\mathrm{d}s}D_{F}\left(\Phi\left(x_{s}\right),\Phi\left(y\right)\right)_{\mid s=1}&=0
\end{align}
we must have 
\begin{align}
D^{F}\left(\Phi\left(x\right),\Phi\left(y\right)\right)
&=\frac{\mathrm{d}^{2}}{\mathrm{d}s^{2}} D_{F} \left(\Phi\left(x_{s}\right),\Phi\left(y\right)\right)_{\mid s=1}
\\
&\leq \frac{\mathrm{d}^{2}}{\mathrm{d}s^{2}}D_{F}\left(x_{s},y\right)_{\mid s=1}\\
&\leq D^{F}\left(x,y\right)\, .
\end{align}
If $D^F$ is monotone then Proposition \ref{prop:lokal} implies that $D_F$ is monotone.
\end{proof}

\begin{thm}Let $\mathcal{K}$ denote a convex body with a sufficient regret function $D_F$ that is monotone under dilations. Then $F$ is $C^2$. In particular $D_F$ is a Bregman divergence.
\end{thm}
\begin{proof}
Since $D_F$ is monotone under dilation we have that $D^F$ is monotone under dilations whenever $D^F$ is defined. Let $y$ be a point where $F$ is differentiable and let $0<r<1$ and $z$ be a point such that $F$ is differentiable in $\left(1-r\right)\cdot z + r\cdot y$. Then
\begin{align}
D^F\left(\left(1-r\right)\cdot z + r\cdot x,\left(1-r\right)\cdot z + r\cdot y\right)&\leq D^F\left(x,y\right)\\
\left< r\cdot x-r\cdot y\left| \mathbf{H}\left(\left(1- r\right)\cdot z+r\cdot y\right)\right| r\cdot x-r\cdot y\right>&
\leq \left<  x- y\left| \mathbf{H}\left(y\right)\right|  x- y\right>\\
r^{2}\cdot\left< x-y\left| \mathbf{H}\left(\left(1- r\right)\cdot z+r\cdot y\right)\right| x- y \right>&
\leq \left<  x- y\left| \mathbf{H}\left(y\right)\right| x- y\right>\\
r^{2}\cdot \mathbf{H}\left(\left(1- r\right)\cdot z+r\cdot y\right)&\leq \mathbf{H}\left(y\right) .
\end{align}
Let $\mathcal{L}_1\subseteq{K}$ denote a ball around $y$ with radius $R_1$ and let $\mathcal{L}_2$ denote a ball around $y$ with radius $R_2 <R_1$. Then for any $w\in\mathcal{L}_2$ there exists a $z\in \mathcal{L}_1$ such that $w=\left(1-r\right)\cdot z +r\cdot y$ where $r\geq 1-\frac{R_2}{R_1}$ implying that
\begin{equation}
\left(1-\frac{R_2}{R_1}\right)^{2}\cdot \mathbf{H}\left(w\right)\leq \mathbf{H}\left(y\right) .
\end{equation}
There also exists a $\tilde{z}\in \mathcal{L}_1$ such that $y=\left(1-\tilde{r}\right)\cdot z +\tilde{r}\cdot w$ where $\tilde{r}\geq 1-\frac{R_2}{R_{1}+R_{2}}$ implying that
\begin{equation}
\left(1-\frac{R_2}{R_{1}+R_2}\right)^{2}\cdot \mathbf{H}\left(y\right)\leq \mathbf{H}\left(w\right) .
\end{equation}
We see that if $R_2$ is small  Then $y\to \mathbf{H}\left(y\right)$ is uniformly continuous on any compact subset of the interior of $\mathcal{K}$ restricted to points where $F$ is twice differentiable. 
Therefore $\mathbf{H}$ has a unique continuous extension to $\mathcal{K}$ and we can use the extension of $\mathbf{H}$ to get an extension of $D^F$.
The last thing we need to prove is that the unique extended function $\mathbf{H}$ actually gives the Hesse matrix in any interior point in $\mathcal{K}$. Let $x,y\in\mathcal{K}$. Introduce $x_r =\left(1-r\right)z+r\cdot x$ and $x_r =\left(1-r\right)z+r\cdot x$. Then 
\begin{multline}
D_{F}\left(x,y\right)-\frac{1}{2}D^{F}\left(x,y\right)\geq D_{F}\left(\left(1-r\right)z+r\cdot x,\left(1-r\right)z+r\cdot y\right)-\frac{1}{2}D^{F}\left(x,y\right)\\
= D_{F}\left(x_r ,y_r \right)-\frac{1}{2}D^{F}\left(x_r ,y_r \right)+\frac{1}{2}D^{F}\left(x_r ,y_r \right)
-\frac{1}{2}D^{F}\left(x,y\right)-\frac{1}{2}D^{F}\left(x,y\right)\\
\geq D_{F}\left(x_r ,y_r \right)-\frac{1}{2}D^{F}\left(x_r ,y_r \right)
+\frac{1}{2}\left<x-y\left| r^2\cdot \mathbf{H}\left(y_r \right)-\mathbf{H}\left(y\right)\right| x-y \right>\\
\geq D_{F}\left(x_r ,y_r \right)-\frac{1}{2}D^{F}\left(x_r ,y_r \right)
-\frac{1}{2}\| r^2\cdot \mathbf{H}\left(\left(1-r\right)z+r\cdot y\right)-\mathbf{H}\left(y\right)\| \|x-y\|^2
\end{multline}
Therefore
\begin{multline}
\frac{D_{F}\left(x,y\right)-\frac{1}{2}D^{F}\left(x,y\right)}{\|x-y\|^2}\geq r^2 \frac{D_{F}\left(x_r ,y_r \right)-\frac{1}{2}D^{F}\left(x_r ,y_r \right)}{\|rx-ry\|^2}\\
-\frac{1}{2}\| r^2\cdot \mathbf{H}\left(\left(1-r\right)z+r\cdot y\right)-\mathbf{H}\left(y\right)\|
\end{multline}
and 
\begin{equation}
\lim_{x\to y}\inf \frac{D_{F}\left(x,y\right)-\frac{1}{2}D^{F}\left(x,y\right)}{\|x-y\|^2}\geq 
-\frac{1}{2}\| r^2\cdot \mathbf{H}\left(\left(1-r\right)z+r\cdot y\right)-\mathbf{H}\left(y\right)\|\, .
\end{equation}
Since this holds for all positive $r<1$ we have
\begin{equation}
\lim_{x\to y}\inf \frac{D_{F}\left(x,y\right)-\frac{1}{2}D^{F}\left(x,y\right)}{\|x-y\|^2}\geq  0 \, .
\end{equation}
One can prove that $\lim\sup$ is less that 0 in the same way.
\end{proof}
\begin{thm}
Assume that $f:[0,1]\to\mathbb{R}$ is a convex symmetric function and that the function $F$ is defined as 
$F\left(\sigma\right)=\mathrm{tr}[f\left(\sigma\right)]$. If the Bregman divergence $D_F$ is monotone 
under dilations then $y\to y^2 \cdot f''(y)$ is an increasing function.
\end{thm}
\begin{proof}
Assume that $D_F$ is monotone under dilations. Let $S:[0,1]\to \mathcal{K}$ denote a section. Then a dilation around $S(0)$ 
commutes with the retraction corresponding to the section $S$. Therefore $D_F$ restricted to $S\left([0,1]\right)$ is 
monotone, so we may without loss of generality assume that the convex body is the interval [0,1].

Then $F$ is $C^2$ and $D^F$ is monotone. 
\begin{align}
D^{F}\left(r\cdot x,r\cdot y\right)&=F''\left(r\cdot y\right) \cdot\left(r\cdot x-r\cdot y\right)^2  \\
&= F''\left(r\cdot y\right)\cdot \left(r\cdot y\right)^{2}\cdot\left(\frac{x}{y}- 1\right)^2 \,.
\end{align}
Therefore $y^{2}\cdot F''(y)$ and $y^{2}\cdot f''(y)$ are increasing.

\end{proof}

\begin{thm}
Let $\mathcal{K}$ denote a convex body of rank 2 with a sufficient and strict regret function $D_F$ that is monotone under dilations. Then $\mathcal{K}$ can be represented by a spin factor.
\end{thm}
\begin{proof}
First we note that $\mathcal{K}$ is a spectral set with a center that we will denote $c$. We will embed $\mathcal{K}$ in a vector space with $c$ as the origin. If $\sigma$ and $\rho$ are points on the boundary and $\lambda\in [0,\nicefrac{1}{2} ]$ then 
\begin{equation}
D_{F}\left(\left(1-\lambda\right)\sigma+\lambda\cdot c,c\right) =D_{F}\left(\left(1-\lambda\right)\rho+\lambda\cdot c,c\right) .
\end{equation}
Therefore 
\begin{equation}\label{eq:lokaldivligning}
D^{F}\left(\sigma ,c\right)=k
\end{equation}
for some constant $k.$ Equation (\ref{eq:lokaldivligning}) can be written in terms of the Hesse matrix as
\begin{equation}
\left\langle \sigma -c \left| \mathbf{H}\left(c\right)\right| \sigma-c\right\rangle = k ,
\end{equation}
 and this is the equation for an ellipsoid. The result follows because any ellipsoid is isomorphic to a ball.
\end{proof}
One easily check that if $f:C^{2}\left([0,1]\right)$ then $F\left(\sigma\right)=\mathrm{tr}[f\left(\sigma\right)]$ defines a $C^2$-function on any spin factor.

\begin{thm}
Assume that $f:C^{3}\left([0,1]\right)$ is a convex symmetric function and that the function $F$ is defined as $F\left(\sigma\right)=\mathrm{tr}[f\left(\sigma\right)]$ on a spin factor. 
If $y\to y^2 f(y)$ is an increasing function then the Bregman divergence $D_F$ is monotone under dilations.
\end{thm}
\begin{proof}

Assume that $y\to y^2 f(y)$ is an increasing function. It is sufficient to prove that $D^F\left(x,y\right)$ is decreasing 
under dilations. Let $x\to \left(1-r\right)z+rx$ denote a dilation around $z$ by a factor of $r\in[0,1]$. Then
\begin{equation}
D^{F}\left(\left(1-r\right)z+rx,\left(1-r\right)z+ry\right)
=r^{2}\left< x-y \left|\mathbf{H}\left(\left(1-r\right)z+ry\right)\right|x-y\right>\, .
\end{equation}
so it is sufficient to prove that $r\to r^{2}\mathbf{H}\left(\left(1-r\right)z+ry\right)$ is an increasing matrix function. 
Since $f$ is $C^3$ we may differentiate with respect to $r$ and we have to prove the inequality
\begin{equation}
2r\mathbf{H}\left(\left(1-r\right)z+ry\right)+r^2\frac{\mathrm{d}}{\mathrm{d}r}\mathbf{H}\left(\left(1-r\right)z+ry\right)\geq 0\, .
\end{equation}
Without loss of generality we may assume $r=1$ so that we have to prove that
\begin{equation}\label{directional}
2\mathbf{H}\left(y\right)+\frac{\mathrm{d}}{\mathrm{d}r}\mathbf{H}\left(\left(1-r\right)z+ry\right)_{\mid r=1}\geq 0 \, .
\end{equation}
If $y=\left(y_{1},y_{2},\dots ,y_{d}\right)$ and $\mathbf{H}=\left(\mathbf{H}_{i,j}\right)$ then 
\begin{align}
\frac{\mathrm{d}}{\mathrm{d}r}\mathbf{H}\left(\left(1-r\right)z+ry\right)_{\mid r=1}
&=\left(\frac{\mathrm{d}}{\mathrm{d}r}\mathbf{H}_{i,j}\left(\left(1-r\right)z+ry\right)\right)_{\mid r=1}\\
&=\left<\nabla \mathbf{H}_{i,j}\left(\left(1-r\right)z+ry\right)\mid y-z \right>_{\mid r=1}\\
&=\left<\nabla \mathbf{H}_{i,j}\left(y\right)\mid y-z \right> \, .
\end{align}
Since inequality (\ref{directional}) is invariant under rotations that leave the center and $y$ invariant the same must be the case for the inequality
\begin{equation}
2\left(\mathbf{H}_{i,j}\right)+\left<\nabla \mathbf{H}_{i,j}\left(y\right)\mid y-z \right> \geq 0 \, ,
\end{equation}
but this inequality is linear in $z$ so we may take the mean under all rotated versions of this inequality. If $\bar{z}$ denotes the mean of rotated versions of $z$ we have to prove that 
\begin{equation}
2\left(H_{i,j}\right)+\left<\nabla \mathbf{H}_{i,j}\left(y\right)\mid y-\bar{z} \right>\geq 0 \, .
\end{equation}
Since $\bar{z}$ is collinear with the $y$ and the center we have reduced the problem to dilations of a one-dimensional spin factor which is covered in Theorem \ref{dilmon}.
\end{proof}

For the Tsallis entropy of order $\alpha$ we have $F\left(x\right)=\frac{x^{\alpha}+\left(1-x\right)^{\alpha}-1}{\alpha-1}$
so that $F''\left(x\right)=\alpha\left(x^{\alpha-2}+\left(1-x\right)^{\alpha-2}\right)$
and
\begin{align}
x^{2}F''\left(x\right) & =x^{2}\alpha\left(x^{\alpha-2}+\left(1-x\right)^{\alpha-2}\right)\\
 & =\alpha\left(x^{\alpha}+x^{2}\left(1-x\right)^{\alpha-2}\right)\, .
\end{align}
The derivative is 
\begin{align}
&\alpha\left(\alpha x^{\alpha-1}+2x\left(1-x\right)^{\alpha-2}  -x^{2}\left(\alpha-2\right)\left(1-x\right)^{\alpha-3}\right)\\
  &=\alpha\left(\alpha x^{\alpha-1}+\left(2x\left(1-x\right)-x^{2}\left(\alpha-2\right)\right)\left(1-x\right)^{\alpha-3}\right)\\
&=\alpha\left(\alpha x^{\alpha-1}+x\left(2-\alpha x\right)\left(1-x\right)^{\alpha-3}\right)\\
   &=\alpha x^{\alpha-1}\left(\alpha+\left(\frac{2}{x}-\alpha\right)\left(\frac{1}{x}-1\right)^{\alpha-3}\right).
\end{align}
Set $z=\frac{1}{x}-1$ so that $x=\frac{1}{z+1}$ which gives
\begin{equation}
\alpha+\left(\frac{2}{x}-\alpha\right)\left(\frac{1}{x}-1\right)^{\alpha-3}=\alpha+\left(2z+2-\alpha\right)z^{\alpha-3}.\label{eq:FortegnAndenAfledt}
\end{equation}
For $\alpha\leq2$ the derivative is always positive. For $\alpha<3$
and $z$ tending to zero the derivative tends to $-\infty$ if $2-\alpha$
is negative so we do not have monotonicity for $2<\alpha<3.$ 

For $\alpha\geq3$ we calculate the derivative in order to determine the minimum.
\begin{equation}
2z^{\alpha-3}+\left(2z+2-\alpha\right)\left(\alpha-3\right)z^{\alpha-4} =0 \, ,
\end{equation}
which has the solution $z=\frac{\alpha-3}{2}$. Plugging this solution the expression in Equation (\ref{eq:FortegnAndenAfledt}) gives the value
\begin{equation}
\alpha+\left(2\cdot\frac{\alpha-3}{2}+2-\alpha\right)\left(\frac{\alpha-3}{2}\right)^{\alpha-3}
=\alpha-\left(\frac{\alpha-3}{2}\right)^{\alpha-3}.
\end{equation}
Numerical calculations show that this function is positive for values
of $\alpha$ between 3 and 6.43779 .

\section{Monotonicity of Bregman divergences on Spin Factors}
A binary system can be represented as the spin factor $JSpin_1$ or as the interval [0,1]. 

\begin{thm}\label{dilmon}Let $F:[0,1]\to \mathbb{R}$ denote a convex and symmetric function. Then $D_F$ is monotone if 
and only if $F\in C^2\left(\left[0,1\right]\right)$ and $y\to y^{2}\cdot F''(y)$ is increasing.
\end{thm}
\begin{proof}
The convex body $\left[0,1\right]$ has the identity and a reflection as the only isomorphisms. 
Any affinity can be decomposed into an isomorphism and two dilations where each dilation is a dilation around one of the 
extreme points $\left \{ 0,1\right \}$. 
Therefore $D_F$ is monotone if and only if it is monotone under dilations. 
\end{proof}

Next we will study monotonicty of Bregman divergences in spin factors $JSpin_d$ for $d\geq 2$.

\begin{lem}\label{ReduktionTil2}
Let $D_{F}$ denote a Bregman divergence on $JSpin_{d}$ where $d\geq2$.
If $D_{F}$ satisfies sufficiency and the restriction to $JSpin_{2}$
is monotone, then $D_{F}$ is monotone on $JSpin_{d}.$
\end{lem}
\begin{proof}
Assume that $D_{F}$ satisfies sufficiency and that the
restriction of $D_{F}$ to $JSpin_{2}$ is monotone. Let $\rho_{1},\sigma\in JSpin_{d}$
and let $\Phi:JSpin_{d}\to JSpin_{d}$ denote a positive trace preserving
affinity. Let $\Delta$ denote the disc spanned of $\rho_{1},\sigma$ and
$\boldsymbol{\frac{1}{2}}.$ Then $\Phi\left(\rho_{1}\right),\Phi\left(\sigma\right)$
and $\Phi\left(\boldsymbol{\frac{1}{2}}\right)$ spans a disc $\tilde{\Delta}$
in $JSpin_{d}.$ The restriction of $\Phi$ to $\Delta$ can be written
as $\Phi_{\mid\Delta}=\Phi_{2}\circ\Phi_{1}$ where $\Phi_{1}$ is
an affinity $\Delta\to\Delta$ and $\Phi_{2}$ is an isomorphism $\Delta\to\tilde{\Delta}.$
Essentially $\Phi_{2}$ maps a great circle into a small circle where the great circle is the boundary of $\Delta$ and the small circle is the boundary of the $\tilde{\Delta}$. According
to our assumptions $\Phi_{1}$ is monotone so it is sufficient to
prove that $\Phi_{2}$ is monotone. 

\begin{figure}[tbh]
\begin{centering}
\begin{tikzpicture}[scale=1.10,line cap=round,line join=round,>=triangle 45,x=1.0cm,y=1.0cm]
\clip(-2.1,-2.1) rectangle (2.1,2.1);
\draw [fill=gray!30!white,fill opacity=0.5, line width=1.2pt] (0.,0.) circle (2.cm);
\draw [line width=1.2pt] (-2.,0.)-- (2.,0.);
\draw [line width=1.2pt] (1.3237037037037047,1.0140740740740726)-- (1.3237037037037047,-1.0140740740740726);

\draw (-0.22,0.56) node[anchor=north west] {$\mathbf{\boldsymbol{\frac{1}{2}}}$};
\draw (0.5,-0.1) node[anchor=north west] {$\sigma$};
\draw (0.85,1.182962962962961) node[anchor=north west] {$\rho_{1}$};
\draw (0.85,-0.8170370370370342) node[anchor=north west] {$\rho_{2}$};
\draw (1.3592592592592603,0) node[anchor=north west] {$\bar{\rho}$};

\draw [fill=blue] (0.,0.) circle (2.0pt);
\draw [fill=blue] (1.3237037037037047,1.0140740740740726) circle (2.0pt);
\draw [fill=blue] (0.6837037037037046,0.) circle (2.0pt);
\draw [fill=blue] (1.3237037037037047,-1.0140740740740726) circle (2.0pt);
\draw [fill=blue] (1.3237037037037047,0.) circle (2.0pt);
\end{tikzpicture}
\par\end{centering}
\caption{Illustration of $\Delta$ and the relative position of the states mentioned in the proof of Lemma \ref{ReduktionTil2}.}
\end{figure}
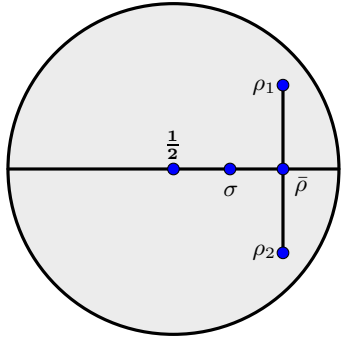

Let $\rho_{2}$ denote a state such that 
\begin{align}
D_{F}\left(\rho_{2},\sigma\right)&=D_{F}\left(\rho_{1},\sigma\right)\\
D_{F}\left(\rho_{2},\boldsymbol{\frac{1}{2}}\right)&=D_{F}\left(\rho_{1},\boldsymbol{\frac{1}{2}}\right).
\end{align}
Then 
\begin{align}
D_{F}\left(\rho_{1},\sigma\right) & =\frac{1}{2}\cdot D_{F}\left(\rho_{1},\sigma\right)+\frac{1}{2}\cdot D_{F}\left(\rho_{2},\sigma\right)\\
 & =\frac{1}{2}\cdot D_{F}\left(\rho_{1},\bar{\rho}\right)+\frac{1}{2}\cdot D_{F}\left(\rho_{2},\bar{\rho}\right)+D_{F}\left(\bar{\rho},\sigma\right)
\end{align}
where $\bar{\rho}=\frac{1}{2}\cdot\rho_{1}+\frac{1}{2}\cdot\rho_{2}.$
Now $\bar{\rho},\sigma$ and $\boldsymbol{\frac{1}{2}}$ are co-linear and so
are $\Phi_{2}\left(\bar{\rho}\right),\Phi_{2}\left(\sigma\right),$
and $\Phi_{2}\left(\boldsymbol{\frac{1}{2}}\right)$ so the restriction of
$\Phi$ to the span of $\bar{\rho},\sigma$ and $\boldsymbol{\frac{1}{2}}$
is an interval and the span of $\Phi_{2}\left(\bar{\rho}\right),\Phi_{2}\left(\sigma\right),\Phi_{2}\left(\boldsymbol{\frac{1}{2}}\right),$
and $\boldsymbol{\frac{1}{2}}$ is a disc so by assumption the restriction
is monotone implying that 
\begin{equation}
D_{F}\left(\Phi_{2}\left(\bar{\rho}\right),\Phi_{2}\left(\sigma\right)\right)\leq D_{F}\left(\bar{\rho},\sigma\right)\,.
\end{equation}

Let $\bar{\pi}\in\Delta$ denote a state that is colinear with $\bar{\rho}$
and $\boldsymbol{\frac{1}{2}}$ and such that $D_{F}\left(\bar{\pi},\boldsymbol{\frac{1}{2}}\right)=D_{F}\left(\Phi_{2}\left(\bar{\rho}\right),\boldsymbol{\frac{1}{2}}\right).$
Then there exists an affinity $\Psi:\Delta\to\Delta$ such that $\Psi\left(\bar{\rho}\right)=\bar{\pi}$
and for $i=1,2$
\begin{equation}
D_{F}\left(\Phi\left(\rho_{i}\right),\Phi\left(\bar{\rho}\right)\right)  =D_{F}\left(\Psi\left(\rho_{i}\right),\Psi\left(\bar{\rho}\right)\right)\,.
\end{equation}
Since $\Psi$ is monotone
\begin{equation}
D_{F}\left(\Phi\left(\rho_{1}\right),\Phi\left(\bar{\rho}\right)\right)=D_{F}\left(\Phi\left(\rho_{2}\right),\Phi\left(\bar{\rho}\right)\right)\leq D_{F}\left(\rho_{1},\bar{\rho}\right).
\end{equation}
 Therefore 
\begin{multline}
D_{F}\left(\Phi\left(\rho_{1}\right),\Phi\left(\sigma\right)\right)  =\frac{1}{2}\cdot D_{F}\left(\Phi\left(\rho_{1}\right),\Phi\left(\sigma\right)\right)+\frac{1}{2}\cdot D_{F}\left(\Phi\left(\rho_{2}\right),\Phi\left(\sigma\right)\right)\\
  =\frac{1}{2}\cdot D_{F}\left(\Phi\left(\rho_{1}\right),\Phi\left(\bar{\rho}\right)\right)+\frac{1}{2}\cdot D_{F}\left(\Phi\left(\rho_{2}\right),\Phi\left(\bar{\rho}\right)\right)+D_{F}\left(\Phi\left(\bar{\rho}\right),\Phi\left(\sigma\right)\right)\\
  \leq\frac{1}{2}\cdot D_{F}\left(\rho_{1},\bar{\rho}\right)+\frac{1}{2}\cdot D_{F}\left(\rho_{2},\bar{\rho}\right)+D_{F}\left(\bar{\rho},\sigma\right)
  =D_{F}\left(\rho_{1},\sigma\right)\,.
\end{multline}
\end{proof}

\begin{thm}\label{thm:InfDivMonoton}
Information divergence is monotone on spin factors.
\end{thm}

\begin{proof}
According to Lemma \ref{ReduktionTil2} we just have to check monotonicity on spin factors of dimension 2, but these are sections of qubits. M{\"u}ller-Hermes and Reeb \cite{Mueller-Hermes2017}
proved that quantum relative entropy is monotone on density matrices
on complex Hilbert spaces. In particular quantum relative entropy
is monotone on qubits. Therefore
information divergence is monotone on any spin factor.
\end{proof}
We will need the following lemma.
\begin{lem}
\label{lem:dilation}Let $\Phi:\mathcal{K}\to\mathcal{K}$ denote
an affinity of a centrally symmetric set into itself. Let $\Psi_{r}$
denote a dilation around the center $c$ with a factor $r\in\left]0,1\right].$
Then $\Psi_{r}\circ\Phi\circ\Psi_{r}^{-1}$ maps $\mathcal{K}$ into
itself.
\end{lem}

\begin{proof}
Embed $\mathcal{K}$ in a vector space $V$ with origin in the center of $\mathcal{K}$. Then $\Phi$ is given
by $\Phi\left(\vec{v}\right)=A\vec{v}+\vec{b}$ and $\Psi_{r}\left(\vec{v}\right)=r\cdot\vec{v}$.
Then 
\begin{align}
\left(\Psi_{r}\circ\Phi\circ\Psi_{r}^{-1}\right)\left(\vec{v}\right) & =r\cdot\left(\boldsymbol{A}\left(\frac{1}{r}\cdot\vec{v}\right)+\vec{b}\right)\\
 & =\boldsymbol{A}\vec{v}+r\cdot\vec{b}\, .
\end{align}
Assume that $\vec{v}\in\mathcal{K}.$ Then $\Phi\left(\vec{v}\right)\in\mathcal{K}$
and $-\Phi\left(-\vec{v}\right)\in\mathcal{K}$. Hence for $\left(1-t\right)\cdot\Phi\left(\vec{v}\right)+t\cdot\left(-\Phi\left(-\vec{v}\right)\right)\in\mathcal{K}.$
Now 
\begin{align}
\left(1-t\right)\cdot\Phi\left(\vec{v}\right)+t\cdot\left(-\Phi\left(-\vec{v}\right)\right) & =\left(1-t\right)\cdot\left(\boldsymbol{A}\vec{v}+\vec{b}\right)+t\cdot\left(-\left(\boldsymbol{A}\left(-\vec{v}\right)+\vec{b}\right)\right)\\
 & =\boldsymbol{A}\vec{v}+\left(1-2t\right)\cdot\vec{b}\, .
\end{align}
For $t=\frac{1-r}{2}$ we get 
\begin{equation}
\left(\Psi_{r}\circ\Phi\circ\Psi_{r}^{-1}\right)\left(\vec{v}\right)=\left(1-t\right)\cdot\Phi\left(\vec{v}\right)+t\cdot\left(-\Phi\left(-\vec{v}\right)\right)\in\mathcal{K}\, ,
\end{equation}
which completes the proof.
\end{proof}

\begin{thm}
If $D_{F}$ is a monotone Bregman divergence on a spin factor and
$F_{r}\left(x\right)=F\left(\left(1-r\right)\cdot \boldsymbol{\frac{1}{2}} +r\cdot x\right)$
then the Bregman divergence $D_{F_{t}}$ is also monotone.
\end{thm}

\begin{proof}
We have 
\begin{align}
D_{F_{r}}\left(\rho,\sigma\right) & =D_{F}\left(\left(1-r\right)\cdot \boldsymbol{\frac{1}{2}}+r\cdot\rho,\left(1-r\right)\cdot \boldsymbol{\frac{1}{2}}+r\cdot\sigma\right)\\
 & =D_{F}\left(\Psi_{r}\left(\rho\right),\Psi_{r}\left(\sigma\right)\right)
\end{align}
where $\Psi_{r}$ denotes a dilation around $\boldsymbol{\frac{1}{2}}$ by a factor $r\in\left]0,1\right]$.
Let $\Phi$ denote an affinity of the state space into itself. Then
according to Lemma \ref{lem:dilation}
\begin{align}
D_{F_{r}} (\Phi\left(\rho\right)  &   ,\Phi\left(\sigma\right) )  =D_{F}\left(\Psi_{r}\left(\Phi\left(\rho\right)\right),\Psi_{r}\left(\Phi\left(\sigma\right)\right)\right)\\
 & =D_{F}\left(
\left(\Psi_{r}\circ\Phi\right)\left(\Psi_{r}^{-1}\circ\Psi_{r}\left(\rho\right)\right),
\left(\Psi_{r}\circ \Phi\right)
\left(\Psi_{r}^{-1}\circ\Psi_{r}\left(\sigma\right)\right)
\right)\\
 & =D_{F}\left(\left(\Psi_{r}\circ\Phi\circ\Psi_{r}^{-1}\right)\left(\Psi_{r}\left(\rho\right)\right),\left(\Psi_{r}\circ\Phi\circ\Psi_{r}^{-1}\right)\left(\Psi_{r}\left(\sigma\right)\right)\right)\\
 & \leq D_{F}\left(\Psi_{r}\left(\rho\right),\Psi_{r}\left(\sigma\right)\right)\\
 & =D_{F_{r}}\left(\rho,\sigma\right)\, ,
\end{align}
which proves the theorem.
\end{proof}

In \cite{Pitrik2015} joint convexity of Bregman divergences on complex density matrices was studied (see also \cite{Virosztek2017}).
\begin{thm}
The Bregman divergence $D_F$ given by $F(x)=\rm{tr}[f(x)]$ is jointly convex
if and only if $f$ has the form 
\begin{equation}
f\left(x\right)=a\left(x\right)+\frac{\gamma}{2}q\left(x\right)+\int_{0}^{\infty}e_{\lambda}\,\mathrm{d}\mu\left(\lambda\right)\label{eq:MatrixEntropyClass}
\end{equation}
 where $a$ is affine and 
\begin{equation}
q\left(x\right)=x^{2}
\end{equation}
 and 
\begin{equation}
e_{\lambda}\left(x\right)=\left(\lambda+x\right)\ln\left(\lambda+x\right)\,.
\end{equation}
\end{thm}
This result is related to the \emph{matrix entropy class} introduced in \cite{Chen2014}
and further studied in \cite{Hansen2015}. 
The function $q$ generates the Bregman divergence $D_{q}\left(\rho,\sigma\right)=\mathrm{tr}\left[\left(\rho-\sigma\right)^{2}\right]$
and the function $e_{\lambda}$ generates the Bregman divergence 
\begin{equation}
D_{e_{\lambda}}\left(\rho,\sigma\right)=D\left(\rho+\lambda\left\Vert \sigma+\lambda\right.\right).
\end{equation}
We note that
\begin{multline}
D\left(\rho+\lambda\left\Vert \sigma+\lambda\right.\right)=\\
\left(1+2\lambda\right)\cdot D\left(\frac{1}{1+2\lambda}\cdot\rho+\frac{2\lambda}{1+2\lambda}\cdot c\left\Vert \frac{1}{1+2\lambda}\cdot\sigma+\frac{2\lambda}{1+2\lambda}\cdot c\right.\right),\label{eq:Mix}
\end{multline}
which implies that $2\lambda\left(1+2\lambda\right)\cdot D_{e_{\lambda}}\left(\rho,\sigma\right)\to\mathrm{tr}\left[\left(\rho-\sigma\right)^{2}\right]$
so the Bregman divergence $D_{2}$ may be considered as a limiting
case. Now
\begin{equation}
D_{f}\left(\rho,\sigma\right)=\frac{\gamma}{2}\mathrm{tr}\left[\left(\rho-\sigma\right)^{2}\right]+\int_{0}^{\infty}D\left(\rho+\lambda\left\Vert \sigma+\lambda\right.\right)\,\mathrm{d}\mu\left(\lambda\right).\label{eq:integral}
\end{equation}
Note that the Bregman divergence of order $\alpha$ can be written in this way for $\alpha\in \left[1,2\right]$ .

\begin{thm}
Any Bregman divergence based on a function of the form (\ref{eq:MatrixEntropyClass}) is monotone on
spin factors.
\end{thm}

\begin{proof}
The result follows from Equation (\ref{eq:Mix}) and Equation (\ref{eq:integral})
in combination with Theorem \ref{thm:InfDivMonoton}.
\end{proof}

\section{Strict monotonicity}
\begin{defn}
 We say that a regret function is \emph{strictly monotone} if 
\begin{equation}
D_{F}\left(\Phi\left(\rho\right),\Phi\left(\sigma\right)\right)=D_{F}\left(\rho,\sigma\right)
\end{equation}
implies that $\Phi$ is sufficient for $\rho ,\sigma$ . 
\end{defn}
 In \cite{Harremoes2017a} it was proved that strict monotonicity implies monotonicity.  As we shall see in Theorem \ref{thm:strict_monoton} on convex bodies of rank 2 strictness and monotonicity is equivalent to strict monotonicity as long as the Bregman divergence is based on an analytic function.
\begin{lem}
\label{lem:identitet}
Let $\sigma$ denote a point in a convex body
$\mathcal{K}$ with a monotone Bregman divergence $D_{F}.$ If $\Phi:\mathcal{K}\to\mathcal{K}$
is an affinity then the set 
\begin{equation}
C=\left\{ \rho\in C\mid D_{F}\left(\Phi\left(\rho\right),\Phi\left(\sigma\right)\right)=D_{F}\left(\rho,\sigma\right)\right\} 
\end{equation}
 is a convex body that contains $\sigma.$ 
\end{lem}

\begin{proof}
Assume that $\rho_{0},\rho_{1}\in C$ and $t\in\left[0,1\right]$
and $\bar{\rho}=\left(1-t\right)\cdot\rho_{0}+t\cdot\rho_{1}.$ Then
according to the Bregman identity
\begin{multline}
\left(1-t\right)\cdot D_{F}\left(\Phi\left(\rho_{0}\right),\Phi\left(\sigma\right)\right)+t\cdot D_{F}\left(\Phi\left(\rho_{1}\right),\Phi\left(\sigma\right)\right)\\
=\left(1-t\right)\cdot D_{F}\left(\Phi\left(\rho_{0}\right),\Phi\left(\bar{\rho}\right)\right)+t\cdot D_{F}\left(\Phi\left(\rho_{1}\right),\Phi\left(\bar{\rho}\right)\right)+D_{F}\left(\Phi\left(\bar{\rho}\right),\Phi\left(\sigma\right)\right)\\
\leq\left(1-t\right)\cdot D_{F}\left(\rho_{0},\bar{\rho}\right)+t\cdot D_{F}\left(\rho_{1},\bar{\rho}\right)+D_{F}\left(\bar{\rho},\sigma\right)\\
=\left(1-t\right)\cdot D_{F}\left(\rho_{0},\sigma\right)+t\cdot D_{F}\left(\rho_{1},\sigma\right).
\end{multline}
Therefore the inequality must hold with equality and 
\begin{equation}
D_{F}\left(\Phi\left(\bar{\rho}\right),\Phi\left(\sigma\right)\right)=D_{F}\left(\bar{\rho},\sigma\right),
\end{equation}which proves the lemma.
\end{proof}

\begin{thm}\label{thm:strict_monoton}
Let $D_{F}$  denote a monotone Bregman divergence that is strict on a spin factor
based on an analytic function $f$. Then $D_{F}$ is strictly monotone.
\end{thm}

\begin{proof}
Assume that $D_F$ is monotone and that
\begin{equation}\label{eq:lighed}
D_{F}\left(\Phi\left(\rho\right),\Phi\left(\sigma\right)\right)=D_{F}\left(\rho,\sigma\right).
\end{equation}
Let $\rho_{0}$ and $\rho_{1}$ denote extreme points 
such that $\rho$ and $\sigma$ lie on the line segment between $\rho_{0}$
and $\rho_{1}.$ 
Lemma \ref{lem:lokal} implies that 
\begin{equation}\label{eq:LokaltStreng}
D^{F}\left(\Phi\left(\rho\right),\Phi\left(\sigma_{t}\right)\right)=D^{F}\left(\rho,\sigma_{t}\right)
\end{equation} 
for all $s\in\left[0,1\right]$ where $\sigma_{s}=\left(1-s\right)\cdot\rho+s\cdot\sigma$ . Since $f$ is assumed to be analytic
the identity (\ref{eq:LokaltStreng}) must hold for all $t$ for which $\left(1-s\right)\cdot\rho+s\cdot\sigma\geq0$
and this set of values of $s$ coincides with set  of values for which
$\left(1-s\right)\cdot\Phi\left(\rho\right)+s\cdot\Phi\left(\sigma\right)\geq 0\, .$ 
The identity (\ref{eq:LokaltStreng}) also holds if $\rho$ is replaced by any point $\rho '$ on the line segment between $\rho_{0}$ 
and $\rho_{1}$ because both sides of Equation (\ref{eq:LokaltStreng}) are quadratic functions in the first variable. 
Using Proposition \ref{prop:lokal} we see that Equation (\ref{eq:lighed}) can be extended to any pair of points on the line segment between $\rho_{0}$ and $\rho_{1}\, .$ In particular
\begin{equation}
D_{F}\left(\Phi\left(\rho_{i}\right),\Phi\left(\bar{\rho}\right)\right)=D_{F}\left(\rho_{i},\bar{\rho}\right)
\end{equation}
for $i=0,1$ and $\bar{\rho}=\frac{1}{2}\cdot\rho_{0}+\frac{1}{2}\cdot\rho_{1}\, .$ Since both $\rho_i$ and $\Phi\left(\rho_i \right)$ are extreme points we have 
\begin{equation}
D_{F}\left(\Phi\left(\rho_{i}\right),\boldsymbol{\frac{1}{2}}\right)=D_{F}\left(\rho_{i},\boldsymbol{\frac{1}{2}}\right)
\end{equation}
we have $D_{F}\left(\bar{\rho},\boldsymbol{\frac{1}{2}}\right)=D_{F}\left(\Phi\left(\bar{\rho}\right),\boldsymbol{\frac{1}{2}}\right)\, .$
Therefore the points $\bar{\rho}$ and $\Phi\left(\bar{\rho}\right)$ have the
same distance to the center $\boldsymbol{\frac{1}{2}}.$ Therefore there exists a rotation $\Psi$ that
maps $\Phi\left(\rho_{i}\right)$ into $\rho_{i}$. Since $\Psi$
is a recovery map of the states $\rho_{i}$ it is also a recovery
map of $\rho$ and $\sigma\, .$
\end{proof}

An affinity in a Hilbert ball has a unique extension to a positive
trace preserving map in the corresponding spin factor. Here we shall
study such maps with respect to existence of recovery maps and with
respect to monotonicity of Bregman divergences. Let $\Phi$ denote
a positive trace preserving map of $JSpin_{d}$ into itself. Then
the adjoint map $\Phi^{*}$ is defined by 
\begin{equation}
\left\langle \Phi^{*}\left(x\right),y\right\rangle =\left\langle x,\Phi\left(y\right)\right\rangle .
\end{equation}
If $\Phi\left(\sigma\right)$ is not singular then we may define 
\begin{equation}
\Psi\left(\rho\right)=\sigma^{\nicefrac{1}{2}}\Phi^{*}\left(\left(\Phi\left(\sigma\right)\right)^{-\nicefrac{1}{2}}\rho\left(\Phi\left(\sigma\right)\right)^{-\nicefrac{1}{2}}\right)\sigma^{\nicefrac{1}{2}}.
\end{equation}
We observe that $\Psi\left(\Phi\left(\sigma\right)\right)=\sigma.$
If $\Phi$ is an isomorphism then $\Phi\left(y\right)=O^{*}yO$ where
$O$ is an orthogonal map on $JSpin_{d}$ as a Hilbert space. Then
\begin{align}
\left\langle \Phi^{*}\left(x\right),y\right\rangle  & =\left\langle x,\Phi\left(y\right)\right\rangle \\
 & =\mathrm{tr}\left[xO^{*}yO\right]\\
 & =\mathrm{tr}\left[OxO^{*}y\right]\\
 & =\left\langle OxO^{*},y\right\rangle 
\end{align}
so that $\Phi^{*}\left(x\right)=OxO^{*}.$ Then
\begin{align}
\Psi\left(\Phi\left(\rho\right)\right) 
& =\sigma^{\nicefrac{1}{2}}\Phi^{*}\left(\Phi\left(\sigma^{-\nicefrac{1}{2}}\right)\Phi\left(\rho\right)\Phi\left(\sigma^{-\nicefrac{1}{2}}\right)\right)\sigma^{\nicefrac{1}{2}}\\
 & =\sigma^{\nicefrac{1}{2}}O\left(\left(O^{*}\left(\sigma^{-\nicefrac{1}{2}}\right)O\right)O^{*}\rho O\left(O^{*}\left(\sigma^{-\nicefrac{1}{2}}\right)O\right)\right)O^{*}\sigma^{\nicefrac{1}{2}}\\
 & =\rho\, .
\end{align}
Therefore $\Psi$ is a recovery map. This formula extends to any $\rho$
for which there exists a recovery map because $\Phi$ is an isomorphism
between two sections of the state space that contain $\rho$ and $\sigma\, .$


\subsubsection{Acknowledgement} I would like to thank Howard Barnum for pointing my attention to the
notion of pairs of sections and retractions that proved to be very useful in
stating and proving results on this topic. I would also like to thank two anonymous reviewers for their careful reading and their useful comments.


\begin{thebibliography}{10}
\providecommand{\url}[1]{{#1}}
\providecommand{\urlprefix}{URL }
\expandafter\ifx\csname urlstyle\endcsname\relax
  \providecommand{\doi}[1]{DOI~\discretionary{}{}{}#1}\else
  \providecommand{\doi}{DOI~\discretionary{}{}{}\begingroup
  \urlstyle{rm}\Url}\fi

\bibitem{Alexandrov1939}
Alexandrov, A.D.: Almost everywhere existence of the second differential of a
  convex function and some properties of convex surfaces connected with it.
\newblock Leningrad State Univ. Ann. [Uchenye Zapiski] \textbf{6}(335) (1939)

\bibitem{Alfsen2003}
Alfsen, E.M., Schulz, F.W.: Geometry of State Spaces of Operator Algebras.
\newblock Birkh{\"a}user, Boston (2003)

\bibitem{Banerjee2005}
Banerjee, A., Merugu, S., Dhillon, I.S., Ghosh, J.: Clustering with {B}regman
  divergences.
\newblock Journal of Machine Learning Research \textbf{6}, 1705--1749 (2005).
\newblock \urlprefix\url{https://doi.org/10.1137/1.9781611972740.22}

\bibitem{Barnum2015}
Barnum, H., Barret, J., Krumm, M., M{\"u}ller, M.P.: Entropy, majorization and
  thermodynamics in general probabilistic theories.
\newblock In: C.~Heunen, P.~Selinger, J.~Vicary (eds.) Proceedings of the 12th
  International Workshop on Quantum Physics and Logic, \emph{Electronic
  Proceedings in Theoretical Computer Science}, vol. 195, pp. 43--58 (2015).
\newblock \urlprefix\url{https://arxiv.org/pdf/1508.03107.pdf}

\bibitem{Bregman1967}
Bregman, L.M.: The relaxation method of finding the common point of convex sets
  and its application to the solution of problems in convex programming.
\newblock USSR Comput. Math. and Math. Phys. \textbf{7}, 200--217 (1967).
\newblock Translated from Russian

\bibitem{Chen2014}
Chen, R.Y., Tropp, J.: Subadditivity of matrix $\phi$-entropy and concentration
  of random matrices.
\newblock Electron. J. Probab. \textbf{19}(paper 27), 1--30 (2014).
\newblock \urlprefix\url{https://doi.org/10.1214/EJP.v19-2964}

\bibitem{Hansen2015}
Hansen, F., Zhang, Z.: Characterisation of matrix entropies.
\newblock Letters in Mathematical Physics \textbf{105}(10), 1399--1411 (2015).
\newblock \urlprefix\url{https://doi.org/10.1007/s11005-015-0784-8}

\bibitem{Harremoes2017}
Harremo{\"e}s, P.: Divergence and sufficiency for convex optimization.
\newblock Entropy \textbf{19}(5), Article no. 206 (2017).
\newblock \urlprefix\url{https://doi.org/10.3390/e19050206}

\bibitem{Harremoes2017b}
Harremo{\"e}s, P.: Maximum entropy and sufficiency.
\newblock AIP Conference Proceedings \textbf{1853}(1), 040,001 (2017).
\newblock \urlprefix\url{https://doi.org/10.1063/1.4985352}

\bibitem{Harremoes2017a}
Harremo{\"e}s, P.: Quantum information on spectral sets.
\newblock In: 2017 IEEE International Symposium on Information Theory, pp.
  1549--1553 (2017).
\newblock \urlprefix\url{https://doi.org/978-1-5090-4096-4/17/$31.00}

\bibitem{Harremoes2007a}
Harremo{\"e}s, P., Tishby, N.: The information bottleneck revisited or how to
  choose a good distortion measure.
\newblock In: 2007 IEEE International Symposium on Information Theory, pp.
  566--570. IEEE Information Theory Society (2007).
\newblock \urlprefix\url{https://doi.org/10.1109/ISIT.2007.4557285}

\bibitem{Hayashi2016}
Hayashi, M.: Quantum Information Theory: Mathematical Foundation.
\newblock Springer (2016)

\bibitem{Holevo1982}
Holevo, A.S.: Probabilistic and Statistical Aspects of Quantum Theory,
  \emph{North-Holland Series in Statistics and Probability}, vol.~1.
\newblock North-Holland, Amsterdam (1982)

\bibitem{Jencova2006}
Jen\v{c}ov{\'a}, A., Petz, D.: Sufficiency in quantum statistical inference: A
  survey with examples.
\newblock Infinite Dimensional Analysis, Quantum Probability and Related Topics
  \textbf{09}(03), 331--351 (2006).
\newblock \urlprefix\url{https://doi.org/10.1142/S0219025706002408}

\bibitem{Jiao2014}
Jiao, J., Courtade, T., No, A., Venkat, K., Weissman, T.: Information measures:
  the curious case of the binary alphabet.
\newblock IEEE Trans. Inform. Theory \textbf{60}(12), 7616--7626 (2014).
\newblock \urlprefix\url{https://doi.org/10.1109/TIT.2014.2360184}

\bibitem{McCrimmon2004}
{McC}rimmon, K.: A Taste of Jordan Algebras.
\newblock Springer (2004)

\bibitem{Mueller-Hermes2017}
M{\"u}ller-Hermes, A., Reeb, D.: Monotonicity of the quantum relative entropy
  under positive maps.
\newblock Annales Henri Poincar{\'e} \textbf{18}(5), 1777--1788 (2017).
\newblock \urlprefix\url{https://doi.org/10.1007/s00023-017-0550-9}

\bibitem{Petz1988}
Petz, D.: Sufficiency of channels over von {N}eumann algebras.
\newblock Quart. J. Math. Oxford \textbf{39}(1), 97--108, (1988).
\newblock \urlprefix\url{https://doi.org/10.1093/qmath/39.1.97}

\bibitem{Pitrik2015}
Pitrik, J., Virosztek, D.: On the joint convexity of the {B}regman divergence
  of matrices.
\newblock Letters in Mathematical Physics \textbf{105}(5), 675--692 (2015).
\newblock \urlprefix\url{https://doi.org/10.1007/s11005-015-0757-y}

\bibitem{Virosztek2017}
Virosztek, D.: Jointly convex quantum {J}ensen divergences (2017).
\newblock ArXiv: 1712.05324

\end{thebibliography}

\end{document}